\documentclass[twocolumn,preprintnumbers,notitlepage,superscriptaddress,amsmath,10pt,aps,longbibliography]
{revtex4-1}

\usepackage{graphicx}
\usepackage{color}
\usepackage{hyperref}
\usepackage{subfig}
\usepackage{graphicx}
\usepackage{type1cm}
\usepackage{eso-pic}
\usepackage{amsmath}
\usepackage{amssymb}
\usepackage{braket}
\usepackage{gensymb}
\usepackage{bbold}
\usepackage{enumerate}

\usepackage{threeparttable, caption}
\usepackage{xcolor}
\usepackage{array}
\usepackage{float}
\restylefloat{table}
\usepackage{makecell}

\usepackage{amsthm}

\usepackage{placeins}

\theoremstyle{definition}
\newtheorem{definition}{Definition}

\newtheorem{proposition}{Proposition}

\begin{document}

\title{Non-interactive XOR quantum oblivious transfer: optimal protocols and their experimental implementations}

\author{Lara Stroh}
\affiliation{SUPA, Institute of Photonics and Quantum Sciences, School of Engineering and Physical Sciences, Heriot-Watt University, Edinburgh EH14 4AS, United Kingdom}

\author{Nikola Horov\'{a}}
\affiliation{Department of Optics, Faculty of Science, Palack\'{y} University, 17. listopadu 1192/12, 779 00 Olomouc, Czech Republic}

\author{Robert St\'{a}rek}
\affiliation{Department of Optics, Faculty of Science, Palack\'{y} University, 17. listopadu 1192/12, 779 00 Olomouc, Czech Republic}

\author{Ittoop V. Puthoor}
\affiliation{SUPA, Institute of Photonics and Quantum Sciences, School of Engineering and Physical Sciences, Heriot-Watt University, Edinburgh EH14 4AS, United Kingdom}

\author{Michal Mi\v{c}uda}
\affiliation{Department of Optics, Faculty of Science, Palack\'{y} University, 17. listopadu 1192/12, 779 00 Olomouc, Czech Republic}

\author{Miloslav Du\v{s}ek}
\affiliation{Department of Optics, Faculty of Science, Palack\'{y} University, 17. listopadu 1192/12, 779 00 Olomouc, Czech Republic}

\author{Erika Andersson}
\email{E.Andersson@hw.ac.uk}
\affiliation{SUPA, Institute of Photonics and Quantum Sciences, School of Engineering and Physical Sciences, Heriot-Watt University, Edinburgh EH14 4AS, United Kingdom}

\begin{abstract}
Oblivious transfer (OT) is an important cryptographic primitive. Any multi-party computation can be realised with OT as building block. XOR oblivious transfer (XOT) is a variant where the sender Alice has two bits, and a receiver Bob obtains either the first bit, the second bit, or their XOR. Bob should not learn anything more than this, and Alice should not learn what Bob has learnt. Perfect quantum OT with information-theoretic security is known to be impossible. We determine the smallest possible cheating probabilities for unrestricted dishonest parties in non-interactive quantum XOT protocols using symmetric pure states, and present an optimal protocol, which outperforms classical protocols. We also ``reverse" this protocol, so that Bob becomes sender of a quantum state and Alice the receiver who measures it, while still implementing oblivious transfer from Alice to Bob. Cheating probabilities for both parties stay the same as for the unreversed protocol. We optically implemented both the unreversed and the reversed protocols, and cheating strategies, noting that the reversed protocol is easier to implement.
\end{abstract}

\maketitle

\section{Introduction}
Oblivious transfer (OT) is an important cryptographic primitive for two non-trusting parties. It is universal for multi-party computation, i.e., it can be used as a building block to implement any multi-party computation~\cite{Kilian88, Ishai08}. 
1-out-of-2 oblivious transfer (1-2 OT) is probably the most well-known variant, defined by Even \textit{et al.}~\cite{Even85}.
Here, a sender holds two bits and a receiver obtains one of them. The sender should not know which bit the receiver has received and the receiver should only get to know one of the two bits. A few years earlier, oblivious transfer was informally described by Wiesner as a method ``for transmitting two messages either but not both of which may be received"~\cite{Wiesner83}. Other variants of oblivious transfer include Rabin oblivious transfer~\cite{Rabin05}, 1-out-of-$n$ oblivious transfer~\cite{Brassard86}, generalized oblivious transfer~\cite{Brassard03}, and XOR oblivious transfer~\cite{Brassard03}.

Unfortunately, one-sided two-party computation is impossible with information-theoretic security both in the classical and in the quantum setting~\cite{Mayers97, Lo97}. Only with additional restrictions on dishonest parties is perfect quantum oblivious transfer possible; for instance, with bounded quantum storage~\cite{Damgard08} or relativistic constraints~\cite{PitaluaGarcia16, PitaluaGarcia18}.
Another approach to achieve perfect quantum oblivious transfer is to assume that secure bit commitment exists~\cite{CrepeauKilian88}. Bit commitment is impossible with information-theoretic security, both in the classical and in the quantum setting, but commitment protocols with computational security are possible. Assuming bounded quantum storage also makes bit commitment possible, essentially because adversaries can then no longer delay committing, e.g., to a choice of what measurement basis to use.
Nevertheless, cheating probabilities are bounded in quantum oblivious transfer protocols, even if sender and receiver are only constrained by the laws of quantum mechanics. For 1-2 OT, $2/3$ is a general lower bound on the greater of the sender's and the receiver's cheating probabilities~\cite{Chailloux16, Amiri21}. If pure symmetric states are used to represent the sender's bit values, the bound can be increased to $\approx0.749$. By combining a protocol achieving this bound with a trivial classical protocol, cheating probabilities for both sender and receiver can be made equal to $\approx 0.729$~\cite{Amiri21}. This shows that protocols using pure symmetric states are not optimal. However, except for 1-2 OT protocols using pure symmetric states (which are thus known to be suboptimal), there are no known quantum protocols for quantum oblivious transfer where the lower bounds have been proven to be tight.

In this paper, we focus on XOR oblivious transfer (XOT), which is less investigated, but which also is universal for multi-party computation. Here, a sender holds two bits, and a receiver obtains either the first bit, the second bit, or their XOR. As in 1-2 OT, the receiver should not learn anything else and the sender should not know what the receiver has learnt. To our knowledge, it is unclear whether XOT and 1-out-of-2 OT are equivalent in the quantum setting. For imperfect quantum oblivious transfer, one can argue that the ``quantum advantage" is greater for non-interactive quantum XOT protocols than for non-interactive 1-out-of-2 OT protocols (see Section \ref{sec:XOT protocol}).

We also introduce a way of ``reversing" quantum oblivious transfer protocols, so that oblivious transfer can be implemented both ways, while quantum states are still sent from one party to the other party. This is of importance for applications. ``Reversing" the protocol can be understood in terms of a shared entangled state, similar to how one can reimagine~\cite{Ekert91, Bennett92} the original Bennett-Brassard-84 (BB84) protocol~\cite{Bennett84} for quantum key distribution. Unlike for quantum key distribution, though, for oblivious transfer, the two parties do not trust each other. For OT, therefore, it matters who prepares the entangled state; this party could if they so wished prepare a different state. Cheating probabilities can therefore be different in the ``original" and ``reversed" protocols. For our XOT protocol, however, they turn out to be the same.

In Section \ref{sec:SymmetricXOT}, we examine general non-interactive quantum XOT protocols which use pure symmetric states, and give cheating probabilities for the sender and receiver.
In  Section \ref{sec:XOT protocol}, we present an optimal protocol, showing that it achieves lower cheating probabilities than classical XOT protocols. Arguably, the quantum advantage is larger for XOT than for 1-2 OT. The protocol is a practical application of quantum state elimination \cite{BookBarnett09, Crickmore20}, just as the 1-out-of-2 OT protocol in \cite{Amiri21}, since an honest receiver needs to exclude two of the sender's four possible bit combinations.
The XOT protocols are {\em semi-random}~\cite{Amiri21}, meaning that the receiver obtains the sender's first bit, second bit, or their XOR at random. Semi-random and ``standard" XOT protocols are, however, equivalent to each other with classical post-processing; details of this are given in Appendix \ref{app:Equivalence}.

We discuss a ``reversed" version of the protocol in Section \ref{sec:Reversed protocol}, where the sender of the quantum state becomes the receiver of a state, and vice versa, while oblivious transfer is still implemented in the same direction as in the unreversed case.
In the reversed protocol, the receiver similarly obtains their bit values at random. This is again equivalent to a non-random protocol by using classical post-processing; see Appendix \ref{app:Equivalence}.
Finally, in Section~\ref{sec:Experiment}, we present the experimental implementation of both the unreversed and the reversed XOT protocols and the respective optimal cheating scenarios.

\section{Quantum XOT with symmetric states} 
\label{sec:SymmetricXOT}

We consider quantum XOT protocols which satisfy certain properties.
\begin{enumerate}
\item They are non-interactive protocols, where Alice sends Bob a quantum state $\ket{\psi_{x_0 x_1}}$, encoding her bit values $x_0, x_1$, and Bob measures it.
\item Alice's states $\ket{\psi_{x_0 x_1}}$ are pure and symmetric. That is, $\ket{\psi_{01}}=U\ket{\psi_{00}}, \ket{\psi_{11}}=U\ket{\psi_{01}}, \ket{\psi_{10}}=U\ket{\psi_{11}}$, for some unitary $U$ with $U^4=\hat 1$.
\item Each of Alice's bit combinations is chosen with probability $1/4$.
\item When measuring each state $\ket{\psi_{x_0x_1}}$, Bob obtains either $x_0,\ x_1$, or $x_2 = x_0\oplus x_1$ with probability $1/3$.
\end{enumerate}
As for the two first conditions, our results also give lower bounds on cheating probabilities for interactive protocols, where in the last step of the protocol, Bob needs to distinguish between symmetric states.
As for the two last conditions, biased protocols are of course possible but are usually not considered. Any bias can be exploited by cheating parties.

The states $\ket{\psi_{x_0 x_1}}$ need to be chosen so that it is possible for Bob to obtain either $x_0, x_1$, or $x_2 = x_0\oplus x_1$ correctly. We denote an honest Bob's measurement operators by $\Pi_{0*}, \Pi_{1*}, \Pi_{*0}, \Pi_{*1}, \Pi_{\text{XOR}=0}, \Pi_{\text{XOR}=1}$, corresponding to Bob obtaining $x_0=0, x_0=1, x_1=0, x_1=1, x_2=0$, and $x_2=1$, respectively. Bob should obtain either the first or second bit, or their XOR, each with probability $1/3$. The probability of obtaining outcome $m$ is 
\[p_m=\bra{\psi_{jk}}\Pi_{m}\ket{\psi_{jk}},\]
for $m \in \{0*, 1*, *0, *1, \text{XOR}=0, \text{XOR} =1\}$. This probability should be equal to $1/3$ when an outcome is possible and otherwise be equal to $0$. 
In Appendix \ref{app:FGcond}, we derive necessary conditions that Alice's states $\ket{\psi_{x_0 x_1}}$ have to satisfy in order for Bob to be able to correctly obtain either $x_0, x_1$, or $x_2 = x_0\oplus x_1$ with probability $1/3$ each. For example, it has to hold that
\begin{equation*}
|F|\le \frac{1}{3},~~|G|\le \frac{1}{3},
\end{equation*}
where $F$ and $G$ are given by
\begin{eqnarray}
\braket{\psi_{01} | \psi_{00}} &=& \braket{\psi_{11} | \psi_{01}} = \braket{\psi_{10} | \psi_{11}} = \braket{\psi_{00} | \psi_{10}}=F, \nonumber\\
\braket{\psi_{00} | \psi_{11}} &=& \braket{\psi_{01} | \psi_{10}}=G.
\end{eqnarray}
$F$ is in general complex, but since the states are symmetric, $G$ must be real.

Usually, in oblivious transfer, it is assumed that the sender and receiver are choosing their inputs at random. Here, Bob will obtain either $x_0$, $x_1$, or $x_0\oplus x_1$ at random. Using the terminology in~\cite{Amiri21}, we have a semi-random XOR oblivious transfer (XOT) protocol, defined in general as follows.
\begin{definition}[Semi-random XOR oblivious transfer]
Semi-random XOT is a two-party protocol where
\begin{enumerate}
\item Alice chooses her input bits  $(x_0, x_1) \in \{0, 1\}$ uniformly at random, thereby also specifying their XOR $x_2 = x_0 \oplus x_1$, or she chooses \emph{Abort}.
\item Bob outputs the value $b \in \{0, 1, 2\}$ and a bit $y$, or \emph{Abort}.
\item If both parties are honest, then they never abort, $y = x_b$, 
Alice has no information about $b$, and Bob has no information about $x_{(b+1) \text{ mod }3}$ or about $x_{(b+2) \text{ mod }3}$.
\end{enumerate}
\end{definition}
As we show in Appendix \ref{app:Equivalence}, implementing semi-random XOT allows us to realize standard XOT and vice versa, since these two variants of XOT are equivalent up to classical post-processing. That is, classical post-processing can be used to allow Bob to nevertheless make an active (but random from Alice's point of view) choice of whether he receives $x_0, x_1$, or $x_0\oplus x_1$, without changing the cheating probabilities of either party.

There will be a tradeoff between Alice's and Bob's cheating probabilities, as there also is for 1-2 OT~\cite{Chailloux13, Chailloux16, Amiri21}. Broadly speaking, when the states become more distinguishable, Bob's cheating probability increases and Alice's decreases. 
A cheating Bob aims to guess both $x_0$ and $x_1$, which then also implies knowledge of $x_0\oplus x_1$; knowledge of any two bit values implies knowledge of the third one.
Bob can always cheat at least with probability $1/2$ by following the protocol and randomly guessing the bit value(s) he does not obtain. It is standard to define this as ``cheating". In cryptographic protocols, we are often concerned with the probability that a dishonest party will succeed in doing something they are not supposed to do. A random guess of information one does not hold, and subsequent actions using this guessed information, is a cheating strategy which is always possible.
The cheating strategy which maximises Bob's probability to correctly learn both $x_0$ and $x_1$ is evidently a minimum-error measurement. His optimal measurement is a square-root measurement~\cite{Hausladen94, Ban97}, since he wants to distinguish between equiprobable, pure, and symmetric states. 
The square-root measurement has the measurement operators
\begin{equation*}
\Pi_{x_0x_1}=\rho_{\rm ave}^{-1/2}\ket{\psi_{x_0x_1}}\bra{\psi_{x_0x_1}}\rho_{\rm ave}^{-1/2},
\end{equation*}
where $\rho_{\rm ave} = (1/4)\sum_{x_0x_1}\ket{\psi_{x_0x_1}}\bra{\psi_{x_0x_1}}$ is the average density matrix sent to Bob by Alice.
Using, e.g., an approach from~\cite{Dalla15}, Bob's cheating probability can be shown to be 
\begin{align} 
\label{eq:BoundBob}
B_{OT} = \frac{1}{16} \Big|\sqrt{1+G + 2\, \text{Re}\, F} &+ \sqrt{1+G - 2\, \text{Re}\, F} \nonumber \\
+ \sqrt{1-G + 2\, \text{Im}\, F} &+ \sqrt{1-G - 2\, \text{Im}\, F} \Big|^2.
\end{align}
When $F\rightarrow -F$, keeping $G$ the same, Bob's cheating probability is unchanged. For fixed $|F|$ and $|G|$, Bob's cheating probability is minimised for real $F$ if $G\le 0$, and for purely imaginary $F$ if $G\ge 0$. One way to see this is to examine $\sqrt{B_{OT}}$ as a function of $\theta_F$, with $F=|F|e^{i\theta_F}$; it is easy to verify that the maxima and minima of this function are as just described. Broadly speaking, Bob's cheating probability increases with decreasing $|F|$ and $|G|$, which means that the states become more distinguishable. If $F=G=0$, the states are perfectly distinguishable, and Bob's cheating probability is equal to 1. If $| F|=|G|=1$, Bob's cheating probability as given by Eq. \eqref{eq:BoundBob} would be equal to $1/4$. Bob's cheating probability will, however, never be this low, since the states $\ket{\psi_{x_0x_1}}$ have to be chosen so that Bob can obtain one of Alice's bits correctly. 
As mentioned above, even with a random guess, Bob can also always cheat at least with probability $1/2$.
Since it has to hold that $|F|,|G|\le 1/3$ (see Appendix \ref{app:FGcond}), Bob's cheating probability is in fact never lower than $3/4$. This occurs for $F=\pm 1/3, G=-1/3$, and for $F=\pm i/3, G=1/3$.

A cheating Alice aims to guess whether Bob has obtained $x_0$, $x_1$, or $x_2 = x_0\oplus x_1$. 
Even if following the protocol, Alice can always cheat at least with probability $1/3$ with a random guess.
One can consider two different types of protocol. Either Bob does not test the states Alice sends to him, in which case a dishonest Alice can send him any state. Or, similar to the procedure in the 1-2 OT protocol in~\cite{Amiri21}, Alice can send Bob a sequence of states, and Bob asks her to declare what a fraction of them are. Bob then checks that his measurement results agree with Alice's declaration, which can restrict Alice's available cheating strategies. In the latter scenario, it is Alice's average cheating probability which is bounded, instead of her cheating probability for each individual position. Generally, Alice's cheating probability when Bob does not test is at least as high as when he does test.

In the case when Bob tests a fraction of the states she sends him, a dishonest Alice must use an equal superposition of the states she is supposed to send, entangled with a system she keeps on her side, in order to pass Bob's tests. In Appendix \ref{appendix:Bound cheating Alice testing Bob}, we show that when Bob is testing Alice's states, her cheating probability is bounded as
\begin{equation}
\label{eq:AliceCheatGenBound}
    A_{OT} \ge\left\{\begin{array}{c} \frac{1}{3}+\frac{1}{2}|\text{Im}\, F|  +\frac{1}{2}\max(|\text{Re}\, F|, |G|), \text{ if } G\le 0, \\
    \\
     \frac{1}{3}+\frac{1}{2}|\text{Re}\, F|+\frac{1}{2}\max(|\text{Im}\, F|, |G|), \text{ if } G > 0.\end{array}
\right.
\end{equation}

\noindent As expected, the bound on Alice's cheating probability increases with larger $|F|$ and $|G|$. The bound is also unchanged if $F\rightarrow -F$, keeping $G$ the same. Suppose now that we fix $|F|$ and $|G|$. If $G < 0$, we see that we should choose $\text{Im}\, F=0$ in order to minimise the bound on Alice's cheating probability, and if $G>0$, we should choose $\text{Re}\, F=0$. If $G=0$, a real $F$ and a purely imaginary $F$ with the same $|F|$ will give the same bound. As we have already seen, if $|F|$ and $|G|$ are fixed, also Bob's cheating probability is minimised for these same choices of $\theta_F$. The analysis of Alice's cheating probability when Bob is not testing her states below will further confirm that these are the optimal choices of $\theta_F$ for quantum XOT protocols using symmetric pure states. 

While the bound on Alice's cheating probability in equation \eqref{eq:AliceCheatGenBound} increases with $|F|$ and $|G|$, Bob's cheating probability in Eq. \eqref{eq:BoundBob} is larger for smaller $|F|$ and $|G|$. Together, \eqref{eq:BoundBob} and \eqref{eq:AliceCheatGenBound} give a tradeoff relation between Alice's and Bob's cheating probabilities for non-interactive quantum XOT protocols using pure symmetric states when Bob is testing Alice's states.

\begin{figure*}[!htb]
  
\begin{minipage}{.25\linewidth}
\centering
\subfloat(a){\label{Top Plot for G=-1/3}\includegraphics[width=.8\linewidth]{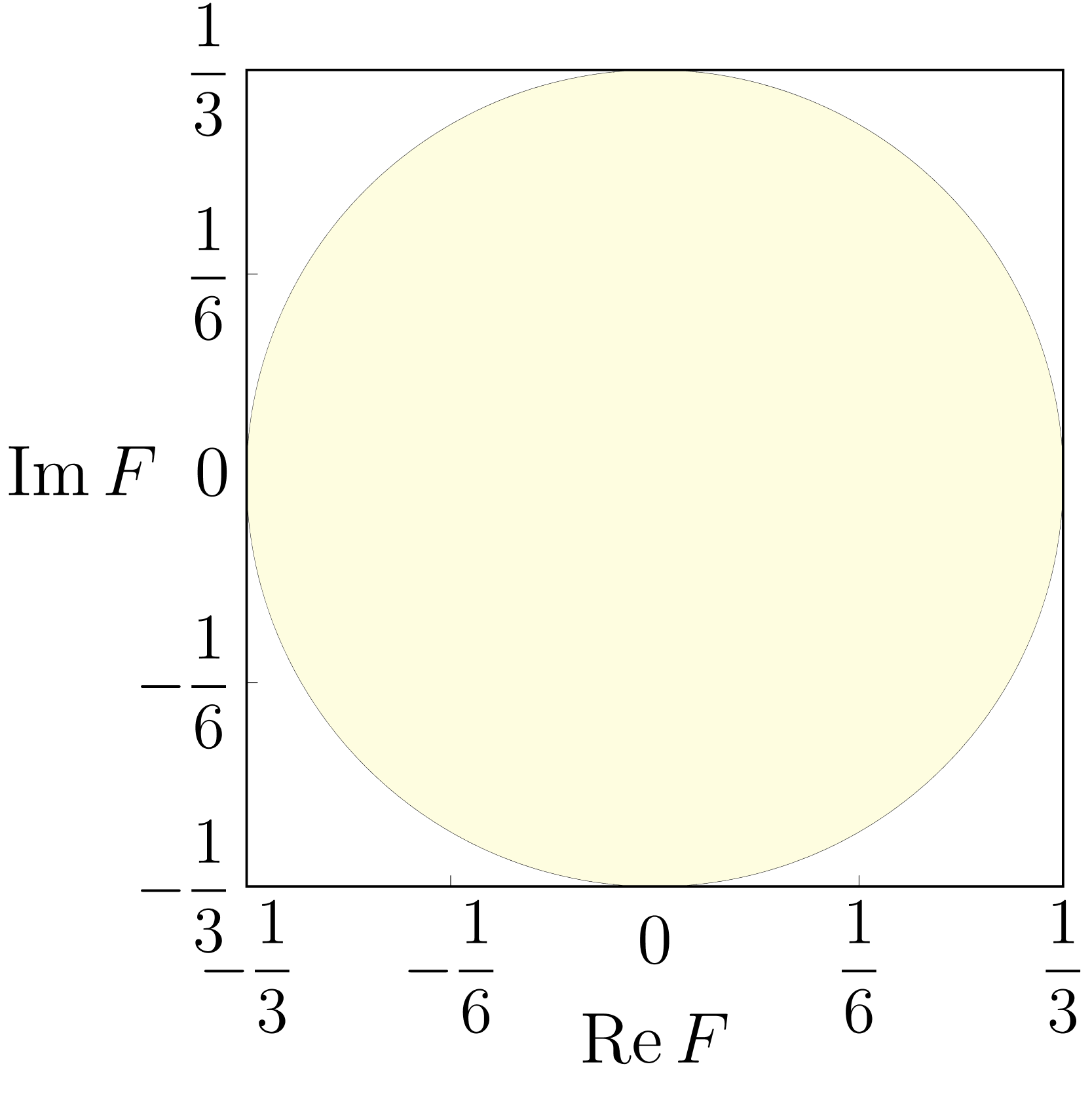}}
\end{minipage}
\begin{minipage}{.25\linewidth}
\centering
\subfloat(b){\label{Top Plot for G=-1/6}\includegraphics[width=.8\linewidth]{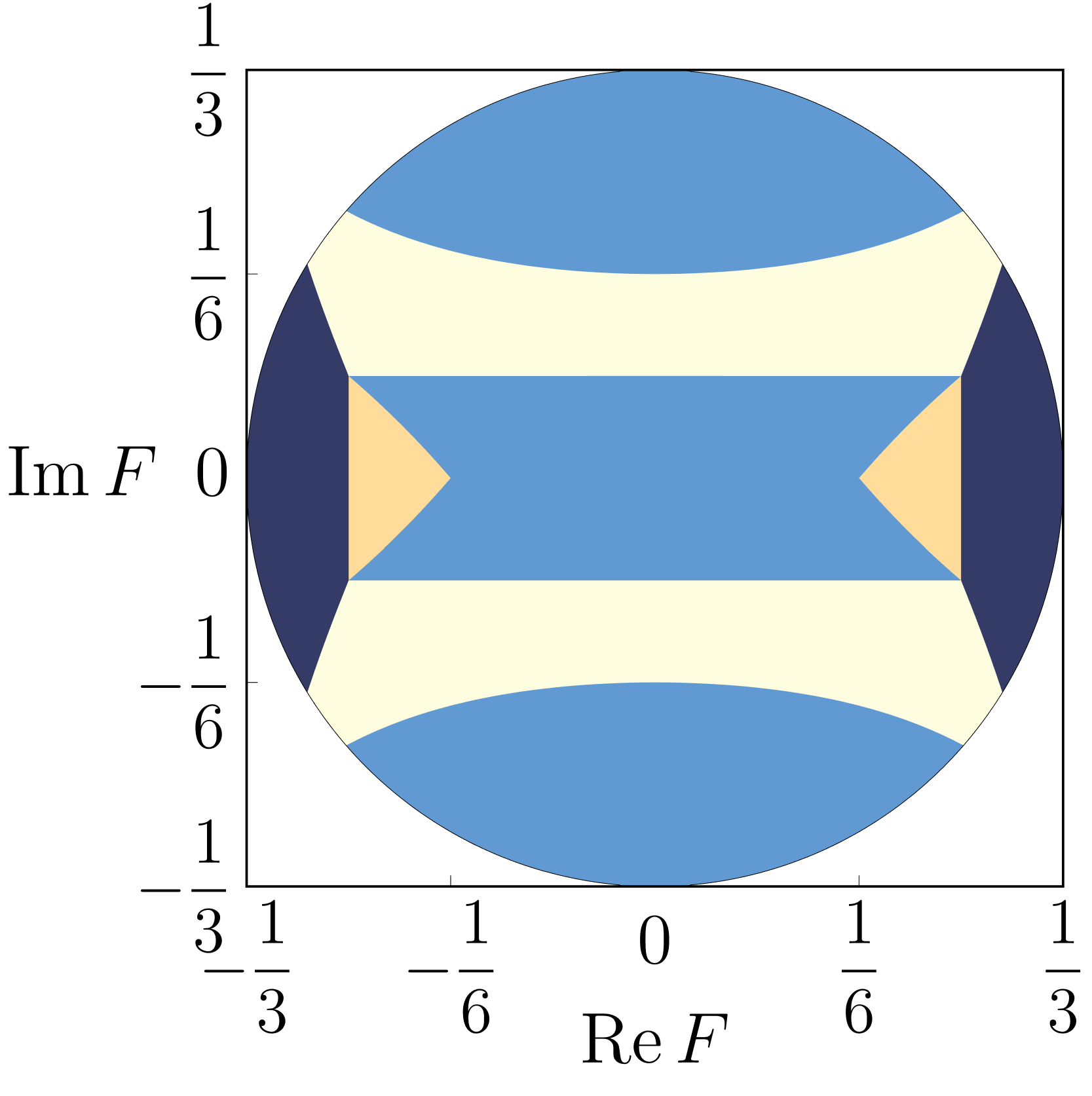}}
\end{minipage}
\begin{minipage}{.25\linewidth}
\centering
\subfloat(c){\label{Top Plot for G=0}\includegraphics[width=.8\linewidth]{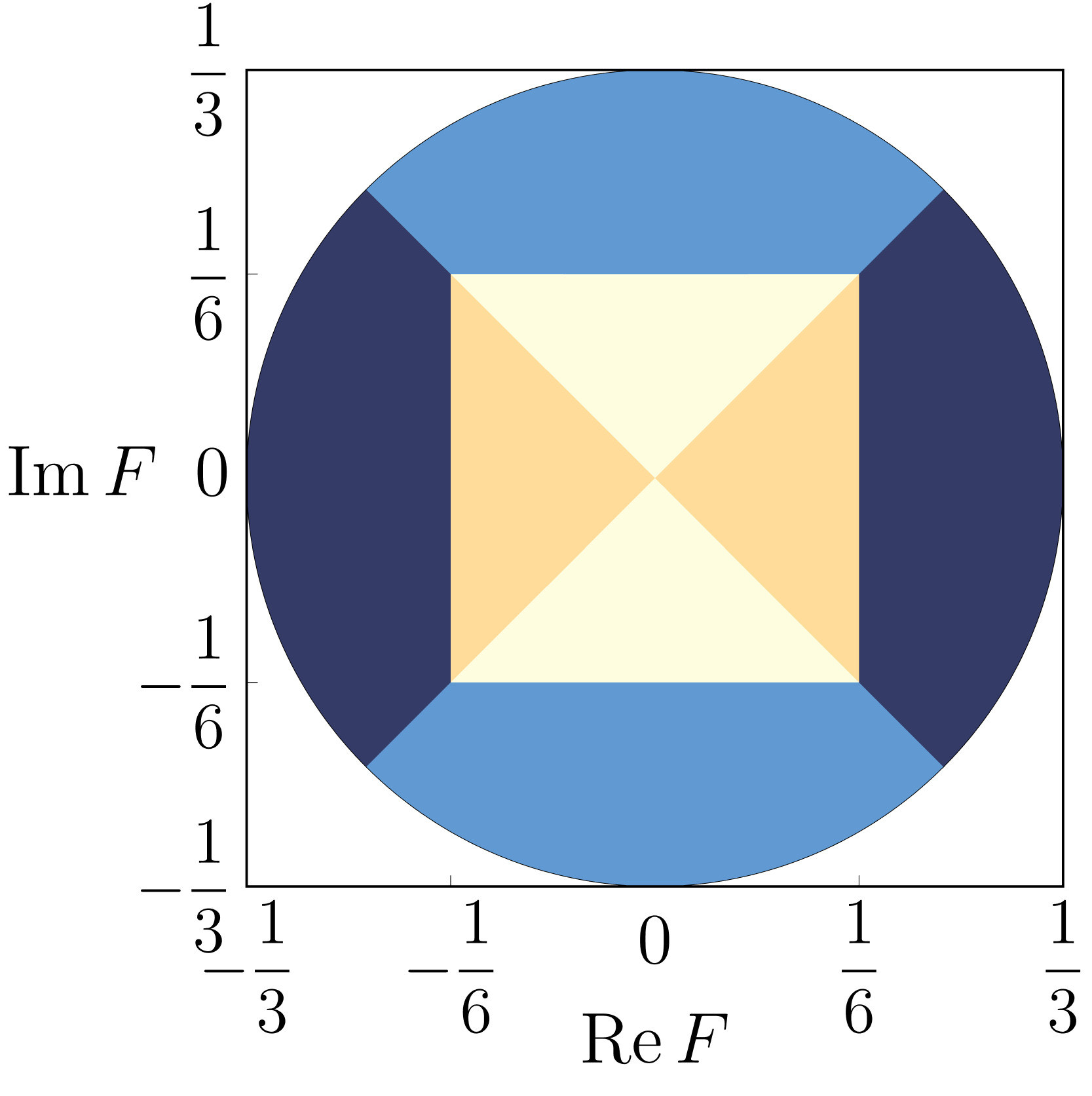}}
\end{minipage}
\begin{minipage}{.2\linewidth}
\centering
\subfloat{\label{Legend for complex F}\includegraphics[width=.5\linewidth]{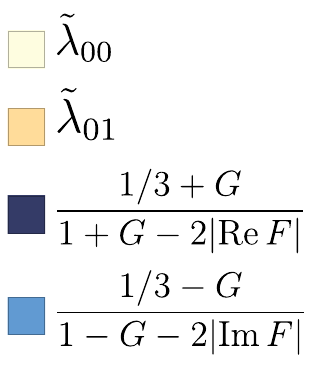}}
\end{minipage}\par\medskip

\begin{minipage}{.3\linewidth}
\centering
\subfloat(d){\label{3D Plot for G=-1/3}\includegraphics[width=\linewidth]{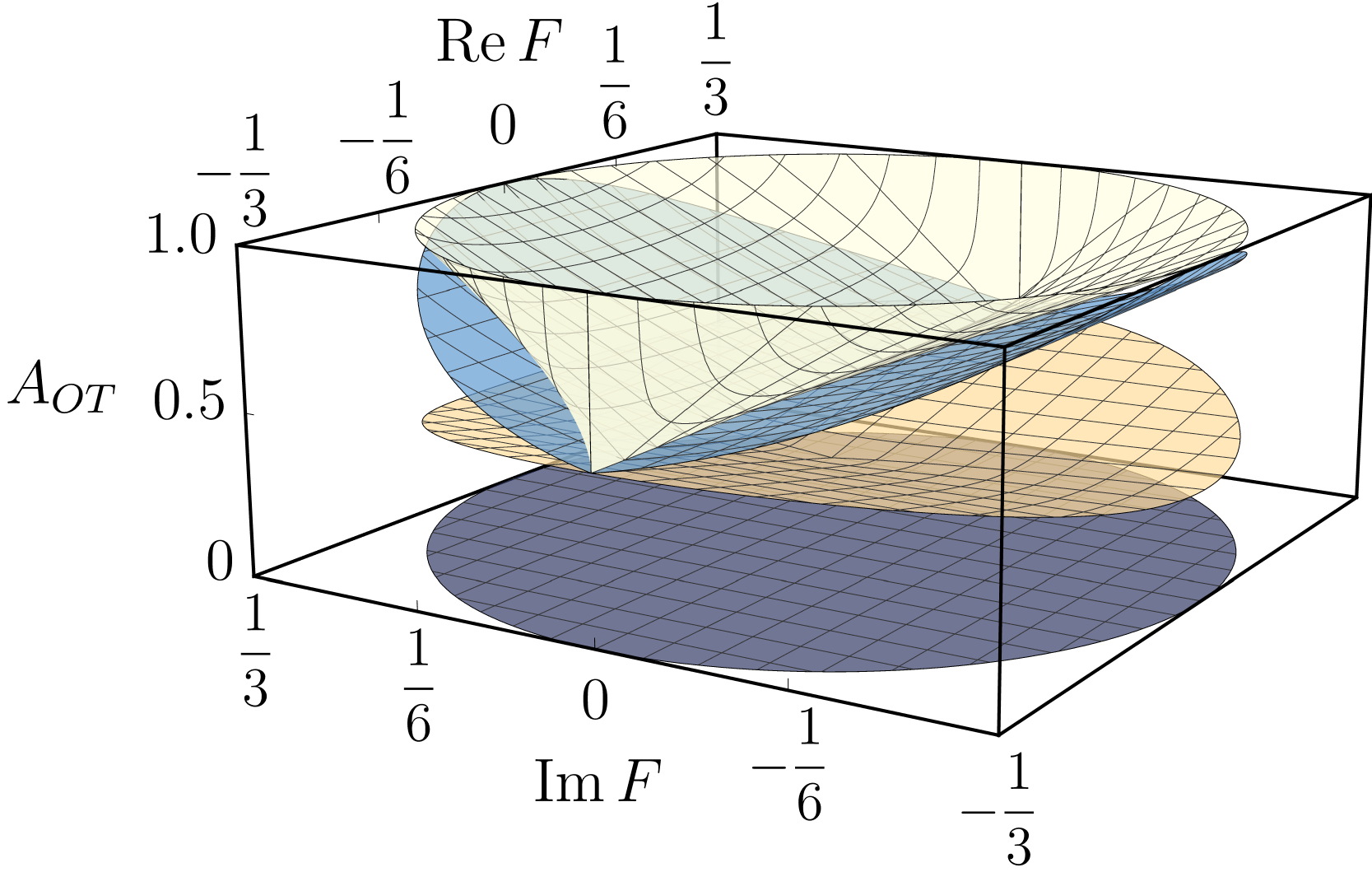}}
\end{minipage}
\hspace{0.3cm}
\begin{minipage}{.3\linewidth}
\centering
\subfloat(e){\label{3D Plot for G=-1/6}\includegraphics[width=\linewidth]{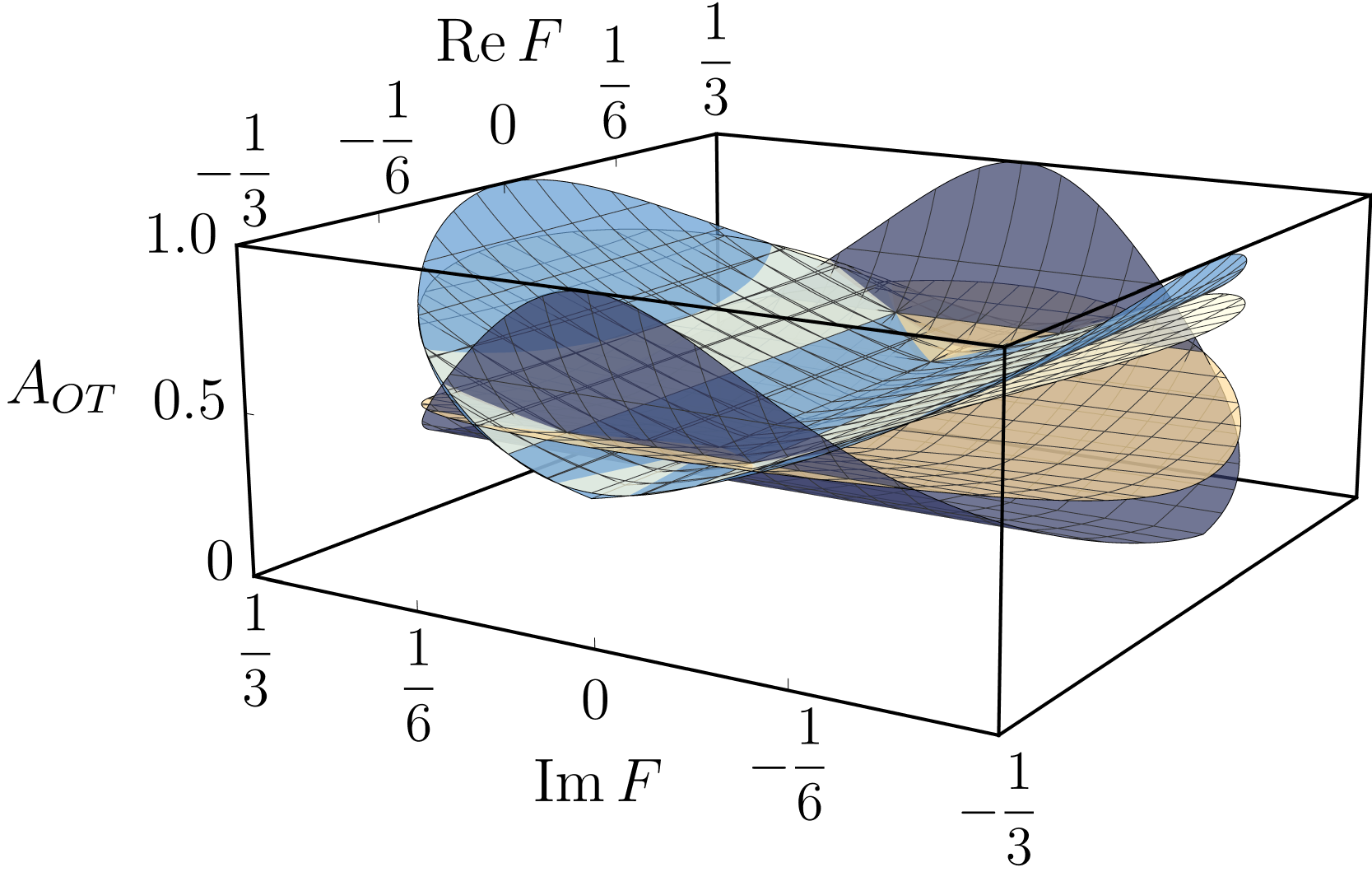}}
\end{minipage}
\hspace{0.3cm}
\begin{minipage}{.3\linewidth}
\centering
\subfloat(f){\label{3D Plot for G=0}\includegraphics[width=\linewidth]{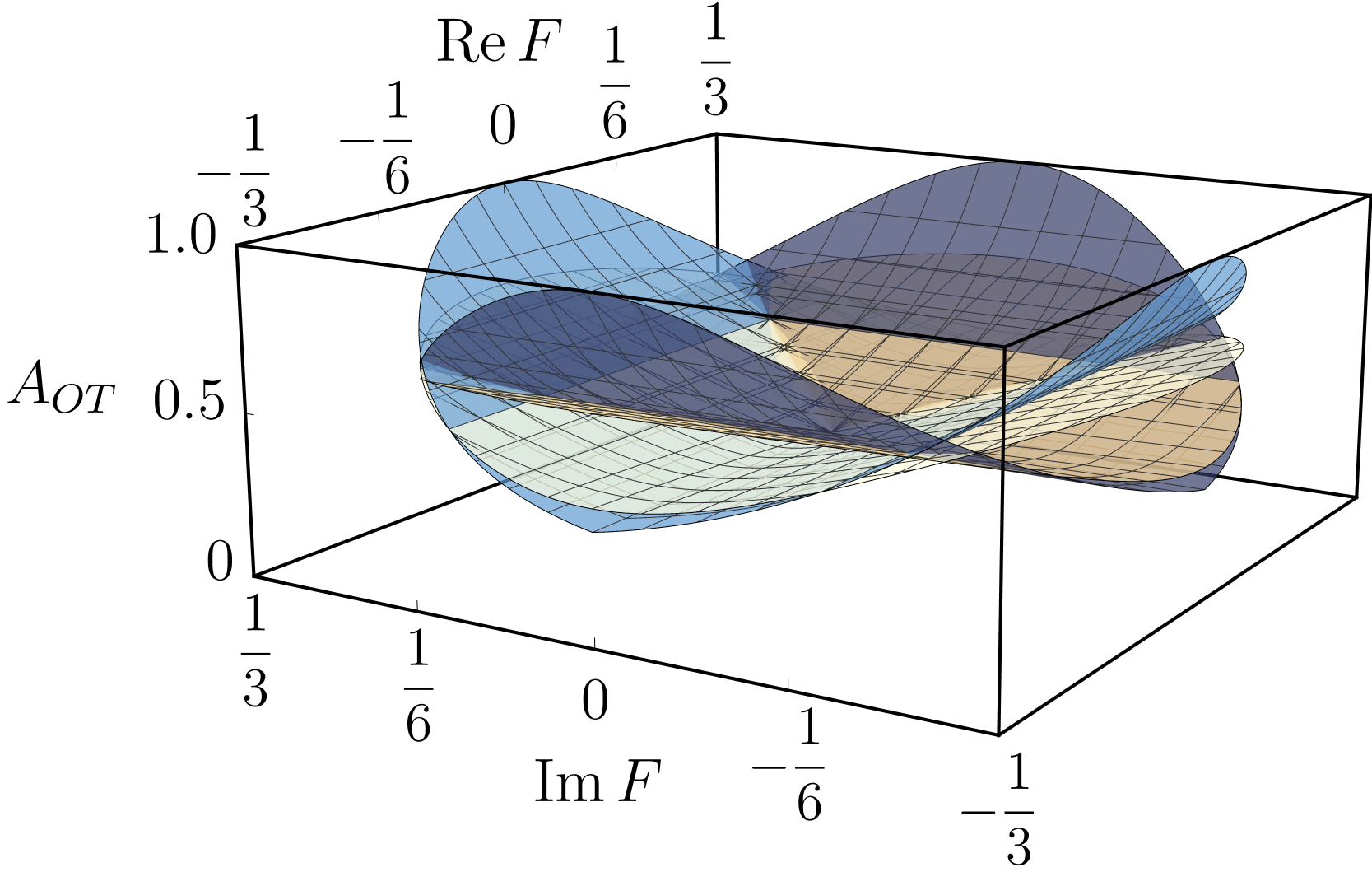}}
\end{minipage}\par\medskip

\caption{Alice's optimal cheating probabilities in Eqns. \eqref{Alice cheats Bob tests p(b=0)} and \eqref{eq:Alice cheats Bob tests p(b=2)} for different values of $G$: (a), (d) $G=-1/3$; (b), (e) $G=-1/6$; (c), (f), $G=0$. The plots (a)–(c) show top-down views of (d)–(f), respectively.}
\label{fig:Plots Alice cheating}
\end{figure*}

When Bob is not testing, it is optimal for Alice to send him the pure state, within the subspace spanned by the states she is supposed to send him, for which Bob's probability to obtain either $x_0$, $x_1$, or $x_2$ is maximised. In Appendix \ref{appendix:Bound cheating Alice no testing Bob}, we show that the largest $p(x_0)$ and $p(x_1)$ Alice can achieve is equal to
\begin{equation}
p(x_0)_{\rm max}=p(x_1)_{\rm max}=\max (\tilde\lambda_{00}, \tilde\lambda_{01}),
\end{equation}
\begin{widetext}
\noindent where
\begin{eqnarray} \label{Alice cheats Bob tests p(b=0)}
&&\tilde\lambda_{00}=\frac{1}{(1+G)^2-4({\rm Re}\, F)^2}\left[\frac{1}{3}(1+G)-2({\rm Re}\, F)^2\right. \left. + \sqrt{\bigg(\frac{1}{3}+G \bigg)^2({\rm Re}\, F)^2+[(1+G)^2-4({\rm Re}\, F)^2]({\rm Im}\, F)^2}\right],\nonumber\\
&&\tilde\lambda_{01}=\frac{1}{(1-G)^2-4({\rm Im}\, F)^2}\left[\frac{1}{3}(1-G)-2({\rm Im}\, F)^2\right.
\left. + \sqrt{\bigg(\frac{1}{3}-G \bigg)^2({\rm Im}\, F)^2+[(1-G)^2-4({\rm Im}\, F)^2]({\rm Re}\, F)^2}\right],
\end{eqnarray}
\end{widetext}
and the largest $p(x_2)$ is
\begin{equation} \label{eq:Alice cheats Bob tests p(b=2)}
p(x_2)_{\rm max}=
\begin{cases}
\frac{1/3+G}{1+G-2|{\rm Re}\, F|} &{\rm if}~~G\ge \frac{|{\rm Im}\, F|-|{\rm Re}\, F|}{2-3|{\rm Re}\, F|-3|{\rm Im}\, F|}\\
\\
\frac{1/3-G}{1-G-2|{\rm Im}\, F|}&{\rm if}~~G< \frac{|{\rm Im}\, F|-|{\rm Re}\, F|}{2-3|{\rm Re}\, F|-3|{\rm Im}\, F|}.
\end{cases}
\end{equation}
Alice's overall cheating probability is then the larger of $p(x_0)_{\rm max}=p(x_1)_{\rm max}$ and $p(x_2)_{\rm max}$. 

The expressions for $p(x_0)$ and $p(x_1)$ in Eq. \eqref{Alice cheats Bob tests p(b=0)} are somewhat complicated, but can be numerically investigated and plotted. In Fig. \ref{fig:Plots Alice cheating}, we plot Alice's cheating probabilities for $G=-1/3$, $G=-1/6$, and $G=0$. One can note that $p(x_i)_{\rm max}$ do not depend on the sign of ${\rm Re}\, F$ and ${\rm Im}\, F$, and that Alice's cheating probabilities are unchanged if ${\rm Re}\, F$ and ${\rm Im}\, F$ are interchanged, with $G$ changing to $-G$.

When $F$ is real, Alice's cheating probabilities reduce to
\begin{eqnarray}
\label{eq:Alicecheatsp01}
p(b=0)_{\max} &=&\nonumber\\
p(b=1)_{\max}&=&\begin{cases}
(i)~\frac{1/3+|F|}{1-G}&{\rm if}~G\ge-|F|\\
\\
(ii)~\frac{1/3+|F|}{1+G+2|F|}&{\rm if}~G\le-|F|
\end{cases}
\end{eqnarray}
and
\begin{equation}
\label{eq:Alicecheatsp2}
p(b=2)_{\max} = \begin{cases}
(iii)~\frac{1/3+G}{1+G-2|F|}&{\rm if}~G\ge\frac{-|F|}{2-3|F|} \\
\\
(iv)~\frac{1/3-G}{1-G}&{\rm if}~ G\le\frac{-|F|}{2-3|F|}.
\end{cases}
\end{equation}
Alice's overall cheating probability is again the largest of the four expressions in Eqns. \eqref{eq:Alicecheatsp01} and \eqref{eq:Alicecheatsp2}. 
For $G\le-|F|$, the largest probability is expression $(iv)$ for $p(b=2)$, and for $G\ge-|F|$, the largest expression is either expression $(i)$ for $p(b=0)_{\max} =p(b=1)_{\max}$ or expression $(iii)$ for $p(b=2)$, depending on $F$ and $G$. Expression $(i)$ is greater than $(iii)$ when $G>|F|$ or $G<1/3-2|F|$. This means that in the region $F>0$, $(iii)$ is the largest expression in the hourglass-shaped region in between the lines $G=F$ and $G=1/3-2F$, and $(i)$ is the largest expression in the two wedges at either side; see Fig. \ref{fig:RealF}. In the region $F<0$, the hourglass and wedges are instead formed by the lines $G=-F$ and $G=1/3+2F$.

\begin{figure*}[!htb]

\begin{minipage}{.3\linewidth}
\centering
\subfloat(a){\label{Region-plot-44-45}\includegraphics[width=.8\linewidth]{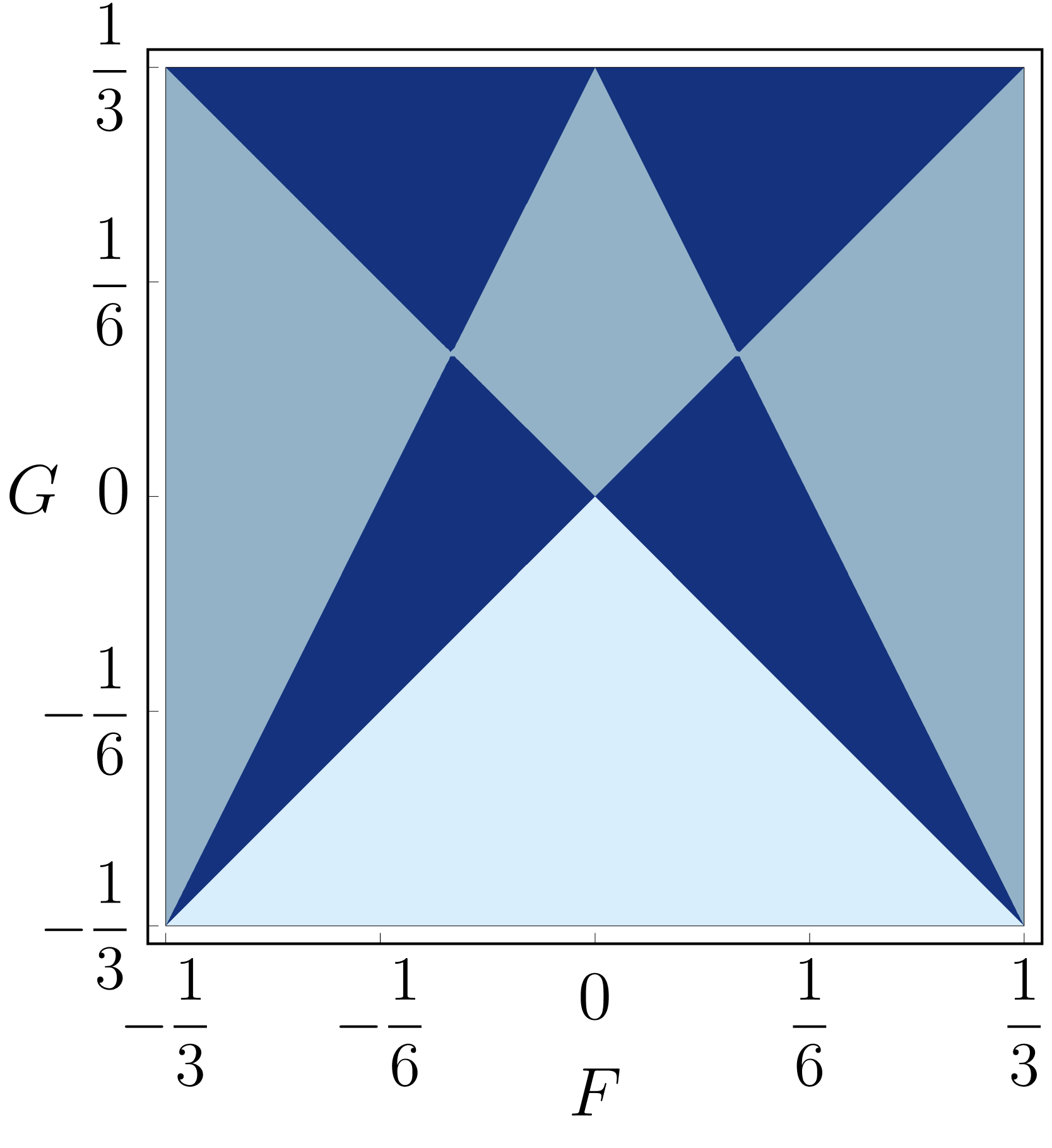}}
\end{minipage}
\hspace{0.7cm}
\begin{minipage}{.3\linewidth}
\centering
\subfloat(b){\label{Region-plot-44-45-3D}\includegraphics[width=1.1\linewidth]{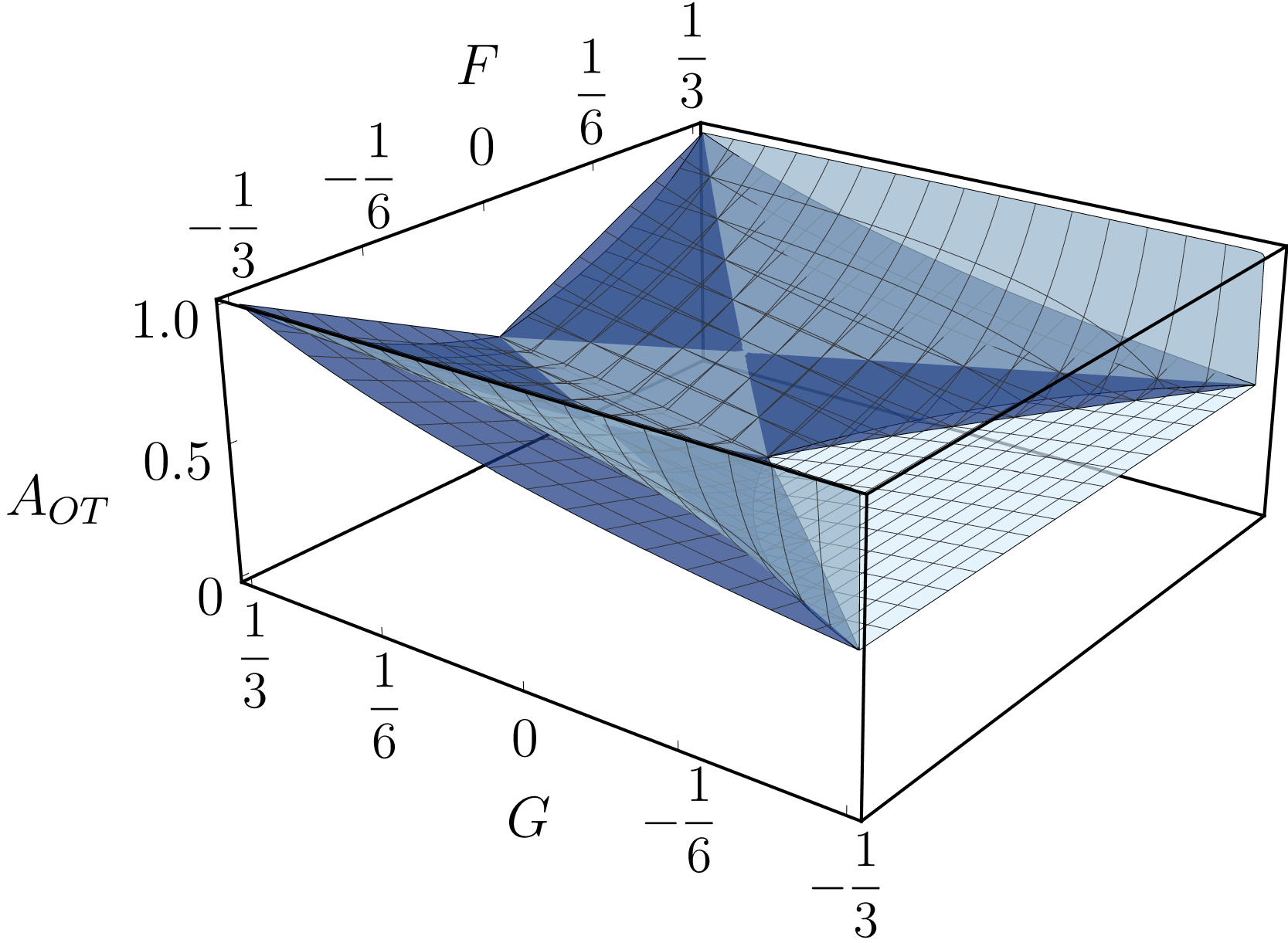}}
\end{minipage}
\hspace{0.7cm}
\begin{minipage}{.3\linewidth}
\centering
\subfloat{\label{Legend for real F}\includegraphics[width=.5\linewidth]{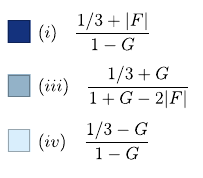}}
\end{minipage}\par\medskip

\caption{Alice's optimal cheating probabilities for real $F$ in Eqns. \eqref{eq:Alicecheatsp01} and \eqref{eq:Alicecheatsp2}. The region plot in (a) is the top view of the 3D plot in (b), both showing the regions for which $(i)$, $(iii)$, or $(iv)$ are the largest, respectively.}
\label{fig:RealF}
\end{figure*}

Broadly speaking, Alice's cheating probability increases when $|F|$ and $|G|$ increase, that is, when the states she sends become less distinguishable to Bob; this was also the case for a cheating Alice when Bob does test.
Bob's cheating probability in Eq. \eqref{eq:BoundBob} is never smaller than $3/4$. This occurs for $F=\pm1/3, G=-1/3$, and for $F=\pm i/3, G=1/3$. In all of these cases, Alice's cheating probability is equal to $1/2$, irrespective of whether Bob tests her states or not. Her cheating probability can be made smaller than $1/2$, at the expense of further increasing Bob's cheating probability, which already is relatively large at $3/4$.
In ~\cite{Osborn22}, Osborn and Sikora present general lower bounds for Alice's and Bob's cheating probabilities in quantum XOT, that is, $B_{OT} \gtrapprox 0.5073$ or $A_{OT} \gtrapprox 0.3382$. It is, however, not known if protocols exist that are tight with those bounds. The expressions in Eqns. \eqref{eq:BoundBob} and \eqref{eq:AliceCheatGenBound}--\eqref{eq:Alice cheats Bob tests p(b=2)} give actual cheating probabilities, not bounds, as a function of the pairwise overlaps between the states an honest Alice sends.

\section{A non-interactive qutrit XOT protocol} 
\label{sec:XOT protocol}

We present a protocol that can thus be said to be optimal among non-interactive protocols using pure symmetric states, since it achieves the smallest possible cheating probability $3/4$ for Bob and the smallest possible cheating probability $1/2$ for Alice, given that Bob's cheating probability is $3/4$.
In our protocol, Alice encodes two bit values $x_0, x_1$ in one of the four non-orthogonal states
\begin{align} \label{Honest Alices States}
\ket{\phi_{x_0 x_1}} = \dfrac{1}{\sqrt{3}}(\ket{0} + (-1)^{x_1} \ket{1} + (-1)^{x_0} \ket{2}).
\end{align}
These states are symmetric, in the sense that $\ket{\phi_{01}} = U\ket{\phi_{00}}$, $\ket{\phi_{11}} = U^2\ket{\phi_{00}}$, and $\ket{\phi_{10}} = U^3 \ket{\phi_{00}}$ for 
\begin{align}
U = \begin{pmatrix}1 & 0 & 0 \\ 0 & 0 & -1 \\ 0 & 1 & 0\end{pmatrix},
\end{align}
for which it holds that $U^4=\hat 1$.
The states $\ket{\phi_{x_0 x_1}}$ are selected so that it is possible to unambiguously exclude two of them, meaning that it is possible to learn either $x_0$, $x_1$, or $x_0\oplus x_1$. Because the states are non-orthogonal, it is not possible to unambiguously determine which single state was received, meaning that it is impossible for Bob to perfectly learn both bits $x_0, x_1$.

After choosing her bits $(x_0, x_1) \in \{0, 1\}$ uniformly at random, Alice sends the respective state to Bob, who makes an unambiguous quantum state elimination measurement to exclude two of the four possible states. There are six different pairs of states he can exclude. Each excluded pair corresponds to learning either $x_0$, $x_1$, or $x_0\oplus x_1$, with either the value $0$ or $1$. To construct Bob's measurement operators, we need six states, each one orthogonal to a pair of states in Eq. \eqref{Honest Alices States}. The measurement operators are then proportional to projectors onto these six states, normalised so that their sum is equal to the identity matrix. For instance, the measurement operator $\Pi_A = (1/4)(\ket{0} + \ket{2})(\bra{0} + \bra{2})$ will exclude the states $\ket{\phi_{11}}$ and $\ket{\phi_{10}}$, so that Bob's outcome bit will be $x_0=0$; similarly for the other operators. Table \ref{tab:BobMeOp} gives the excluded pairs, the corresponding measurement operators, and the deduced output bits for Bob.

\begin{table}[H]
\begin{tabular}{c|c|c}
 Outcome bit & Eliminated states  &  Measurement operator \\
 \hline 
& & \\[-0.25cm]
 $x_0 = 0$ & $\ket{\phi_{11}}$ and $\ket{\phi_{10}} $  &  $\Pi_A = \dfrac{1}{4}(\ket{0} + \ket{2})(\bra{0} + \bra{2})$ \\[0.3cm]
 $x_0 = 1$ & $\ket{\phi_{00}}$ and $\ket{\phi_{01}} $  &  $\Pi_B = \dfrac{1}{4}(\ket{0} - \ket{2})(\bra{0} - \bra{2})$ \\[0.3cm]
 $x_1 = 0$ & $\ket{\phi_{11}}$ and $\ket{\phi_{01}} $  &  $\Pi_C = \dfrac{1}{4}(\ket{0} + \ket{1})(\bra{0} + \bra{1})$ \\[0.3cm]
 $x_1 = 1$ & $\ket{\phi_{00}}$ and $\ket{\phi_{10}} $  &  $\Pi_D = \dfrac{1}{4}(\ket{0} - \ket{1})(\bra{0} - \bra{1})$ \\[0.3cm]
 $x_2 = 0$ & $\ket{\phi_{01}}$ and $\ket{\phi_{10}} $  &  $\Pi_E = \dfrac{1}{4}(\ket{1} + \ket{2})(\bra{1} + \bra{2})$ \\[0.3cm]
  $x_2 = 1$ & $\ket{\phi_{00}}$ and $\ket{\phi_{11}} $  &  $\Pi_F = \dfrac{1}{4}(\ket{1} - \ket{2})(\bra{1} - \bra{2})$ \\[0.25cm]
  \hline
\end{tabular}
\caption{\label{tab:BobMeOp}Bob's measurement operators and outcomes}
\end{table}

To summarise, our  XOT protocol proceeds as follows:
\begin{enumerate}
\item The sender Alice uniformly at random chooses the bits $(x_0, x_1) \in \{0, 1\}$ and sends the corresponding state $\ket{\phi_{x_0x_1}}$ to the receiver Bob.
\item Bob performs an unambiguous state elimination measurement, excluding two of the possible states with certainty, from which he can deduce either $x_0,\ x_1$, or $x_2 = x_0\oplus x_1$.
\end{enumerate}

Bob's and Alice's cheating probabilities are obtained from the expressions in Eqns. \eqref{eq:BoundBob} and \eqref{eq:AliceCheatGenBound}--\eqref{eq:Alice cheats Bob tests p(b=2)}, with $F=1/3$ and $G=-1/3$. A dishonest Bob can cheat with probability $B_{OT}=3/4$.
Alice's cheating probability is $A_{OT} = 1/2$, whether or not Bob tests the states Alice sends.
Our non-interactive XOT protocol has the same cheating probabilities as protocol (3) given by Kundu \textit{et al.}~\cite{Kundu22}. That protocol, however, uses entanglement and is interactive, that is, quantum states are sent back and forth between sender and receiver. Our protocol achieves the same cheating probabilities but is easier to implement, since it is non-interactive and does not require entanglement. In Appendix \ref{app:Reworking}, we show how our protocol can be related to the protocol in~\cite{Kundu22}.

\subsection{Comparison to classical XOT protocols}
\label{sec:ClassComp}

To compare our quantum XOT protocol to classical protocols, we will consider a combination of two trivial classical XOT protocols. In one, Alice can cheat perfectly, and in the other one, Bob can cheat perfectly, similarly to the two ``bad" classical XOT protocols presented in~\cite{Kundu22}.
\medskip \\
\textbf{Protocol 1:}
Alice has the two bits $(x_0, x_1)$,
and chooses to send Bob either $x_0, x_1$, or $x_2 = x_0 \oplus x_1$.
Afterwards Alice ``forgets'' what she has done.
\medskip \\
Obviously, in this protocol, Alice can cheat with probability $1$, while Bob can only cheat with probability $1/2$. This is his probability to correctly guess the value of the bit he did not receive, and the XOR of Alice's bits.
\medskip \\
\textbf{Protocol 2:}
Alice sends all of $(x_0, x_1, x_2 = x_0 \oplus x_1)$ to Bob, who chooses to read one of these bits and discards the other two without looking at them.
\medskip \\
Obviously, Bob can now cheat with probability $1$, since he could read out both $x_0$ and $x_1$. Alice, on the other hand, can only cheat with probability $1/3$ by guessing which bit Bob has chosen to read out.

Protocol 3 is a combination of Protocol 1 and Protocol 2, following~\cite{Chailloux16}. Alice and Bob conduct an unbalanced weak coin flipping protocol. Its outcome will specify which of the two protocols is implemented. The result is effectively
\medskip \\
\textbf{Protocol 3:}
Protocol 1 is implemented with probability $s$ and Protocol 2 is implemented with probability $(1-s)$.
\medskip \\
The cheating probabilities for Alice and Bob in Protocol 3 are given by
\begin{equation}
A^c_{OT} = \dfrac{1}{3} +\dfrac{2}{3} s \quad \text{ and } \quad B^c_{OT} = 1 - \dfrac{1}{2} s .
\end{equation}
Eliminating $s$, we obtain the tradeoff relation
\begin{equation} \label{classicalBound}
3A^c_{OT} + 4B^c_{OT} = 3 \bigg(\dfrac{1}{3} +\dfrac{2}{3} s \bigg) + 4 \bigg(1 - \dfrac{1}{2} s \bigg) = 5.
\end{equation}
If a quantum protocol achieves $3A_{OT} + 4B_{OT}<5$, it has a quantum advantage.
As shown earlier, we have $A_{OT} = 1/2$ and $B_{OT} = 3/4$ in our quantum XOT protocol. 
Thus, we obtain $3 A_{OT} + 4 B_{OT} = 9/2<5$, meaning that it beats the considered classical protocols.

Arguably, the quantum advantage is larger for XOT than for 1-out-of-2 OT, where analogously defined classical protocols satisfy $A_{OT}+B_{OT}=3/2$, and a quantum protocol achieves $A_{OT}+B_{OT}\approx 3/4 +0.729=1.479$~\cite{Amiri21}. Since 1-out-of-2 OT and XOT are different functionalities, we cannot directly compare their tradeoff relations and quantum advantages in cheating probabilities. However, we can  make a reasonable comparison as follows. The LHS in the tradeoff relation in Eq. \eqref{classicalBound} for XOT is $3A^c_{OT} + 4B^c_{OT}$, which would take the maximal value of $7$ if both Alice and Bob could cheat perfectly, $A_{OT}=B_{OT}=1$. The corresponding maximal value in the tradeoff relation $A_{OT}+B_{OT}=3/2$ for ``classical" 1-out-of-2 OT is $2$. It would therefore not be fair to directly compare the difference between $5$ and $9/2$ (between the RHS in Eq. \eqref{classicalBound} and the value $9/2$ achieved for our quantum XOT protocol) with the difference between $3/2$ and $1.479$ (which are the corresponding values for classical and quantum 1-out-of-2 OT). But if we multiply the difference for the XOT protocol by 2 and the difference for the 1-out-of-2 protocol by 7, the comparison can be said to be fair. That is, since $(5-9/2)\times 2 =1 > (3/2-1.479) \times 7 = 0.147$, it is justified to say that the quantum advantage is larger for XOT than for 1-out-of-2 OT. We could also make the comparison instead using the cheating probabilities in ideal protocols. For ideal XOT, we have $A_{OT}=1/3$ and $B_{OT}=1/2$, giving $3A_{OT}+4B_{OT}=3$ as the RHS of the tradeoff relation. For an ideal 1-out-of-2 OT protocol, we have $A_{OT}=B_{OT}=1/2$ and $A_{OT}+B_{OT}=1$. The ``scaled" quantum advantages then become $5-9/2 = 1/2$ for XOT, which is greater than $(3/2-1.479) \times 3 = 0.063$ for 1-out-of-2 OT.

We also note that the bounds on cheating probabilities for XOT hold for every individual implementation of the protocol, while for the 1-out-of-2 OT protocol in~\cite{Amiri21}, the bound is only for the average cheating probability. In particular, the sender could cheat perfectly in any individual 1-out-of-2 OT round, with a negligible probability of being caught, as long as they cheat only in a sufficiently small number of rounds. In this sense too, XOT gives a greater quantum advantage than the 1-out-of-2 quantum OT protocol in~\cite{Amiri21}.

\section{``Reversing" the XOT protocol} 
\label{sec:Reversed protocol}

It is useful to be able to implement oblivious transfer between two parties ``both ways". One party might only be able to prepare and send quantum states, and the other party might only be able to detect them. This was also noted in~\cite{Crepeau91}. We now show that it is possible to ``reverse" our non-interactive protocol, so that XOR oblivious transfer from Bob to Alice can be achieved by Alice sending quantum states to Bob, who measures them. Alternatively, XOR oblivious transfer from Alice to Bob can be implemented by Bob sending quantum states to Alice, who measures them. Such a ``reversal" of a non-interactive quantum OT protocol, using the procedure we describe, is generally possible, but cheating probabilities may be different in the ``original" and ``reversed" protocols. For our specific XOT protocol, we show that they nevertheless end up being the same.

We will consider non-interactive XOT from Alice to Bob, implemented so that Bob sends Alice one of six states depending on his randomly chosen $x_0, x_1$, or $x_2=x_0\oplus x_1$ and its value. Alice learns $x_0$ and $x_1$ by performing a measurement on Bob's state. 
For the reversed non-interactive XOT protocol, Bob's measurement operators given in Table \ref{tab:BobMeOp} become his states, when normalized to 1, and Alice's states given in Eq. \eqref{Honest Alices States} become her measurement operators, when renormalised so that they sum to the identity operator. The XOT protocol is then performed as follows.
\begin{enumerate}
\item Bob randomly chooses one of the six states
\begin{align}
\ket{\phi_{x_0=0}} =& \dfrac{1}{\sqrt{2}} (\ket{0} + \ket{2}),~~ 
\ket{\phi_{x_0=1}} = \dfrac{1}{\sqrt{2}} (\ket{0} - \ket{2}), \nonumber \\
\ket{\phi_{x_1=0}} =& \dfrac{1}{\sqrt{2}} (\ket{0} + \ket{1}),~~ 
\ket{\phi_{x_1=1}} = \dfrac{1}{\sqrt{2}} (\ket{0} - \ket{1}), \nonumber  \\
\ket{\phi_{x_2=0}} =& \dfrac{1}{\sqrt{2}} (\ket{1} + \ket{2}),~~ 
\ket{\phi_{x_2=1}} = \dfrac{1}{\sqrt{2}} (\ket{1} - \ket{2})
\label{eq:Bobreversedstates}
\end{align}
and sends it to Alice. This choice determines the values of $b$ and bit $x_b$, i.e. Bob's input and output in ``standard" non-random XOT.
\item Alice performs a measurement on the state she has received from Bob, learning the bit values $(x_0, x_1)$. Her measurement operators $\Pi_{x_0 x_1}$ are
\begin{align}
\Pi_{00} &= \dfrac{1}{4} (\ket{0} + \ket{1} + \ket{2})(\bra{0} + \bra{1} + \bra{2}) , \nonumber \\ 
\Pi_{01} &= \dfrac{1}{4} (\ket{0} - \ket{1} + \ket{2})(\bra{0} - \bra{1} + \bra{2}) , \nonumber \\
\Pi_{11} &= \dfrac{1}{4} (\ket{0} - \ket{1} - \ket{2})(\bra{0} - \bra{1} - \bra{2}) , \nonumber \\ 
\Pi_{10} &= \dfrac{1}{4} (\ket{0} + \ket{1} - \ket{2})(\bra{0} + \bra{1} - \bra{2}).
\label{eq:Alicereversemeas}
\end{align}
\end{enumerate}
In terms of $x_0$ and $x_1$, Alice's measurement operators can be written $\Pi_{x_0 x_1} = \ket{\Phi_{x_0 x_1}} \bra{\Phi_{x_0 x_1}}$, where $\ket{\Phi_{x_0 x_1}} = (1/2)(\ket{0} + (-1)^{x_1} \ket{1} + (-1)^{x_0} \ket{2})$.
As in the unreversed XOT protocol, when both parties act honestly, Alice will have two bits, but will not know whether Bob knows her first bit, her second bit, or their XOR. Bob will have one of $x_0, x_1$, or  $x_2 = x_0\oplus x_1$, but will not know anything else, since he can only deduce one bit of information with certainty, based on the state he has sent (if he is honest).

In the reversed XOT protocol, Alice cannot choose her bit values, whereas in the unreversed protocol, Bob could not directly choose $b$. However, as for the unreversed protocol, it is possible to add classical post-processing to turn the reversed protocol into ``standard" non-random XOT, where Alice can choose her bit values, and where Bob can consequently only choose $b$, but not the value of the bit he obtains. In Appendix~\ref{app:Equivalence}, we describe this classical post-processing as well.

The aim of cheating parties in the reversed protocol stays unchanged, i.e. dishonest Alice still wants to learn which output Bob has obtained and dishonest Bob still wants to learn not just one bit but any two of $x_0,\ x_1$, or $x_0\oplus x_1$. In the reversed protocol, he wants to know exactly which of the four two-bit combinations Alice has obtained. 
The cheating probabilities are derived in Appendix \ref{app:ReversedCheating}.
Alice cheats by distinguishing between the three mixed states obtained by pairing up the states in \eqref{eq:Bobreversedstates} that correspond to the same output. A minimum-error measurement gives her a cheating probability of $A^r_{OT} = 1/2$.
As in the unreversed protocol, for the cheating sender of the state (now Bob), there are two scenarios: one, where the receiver of the state (now Alice) tests the state, and another, where she does not test. Also here, Bob's cheating probability for both scenarios is the same, that is, we have $B^r_{OT} = 3/4$ in both cases. When Alice is not testing, Bob can achieve this by sending an eigenvector corresponding to the largest eigenvalue of one of Alice's measurement operators.
If Alice does test, however, he needs to send a superposition of the states he is supposed to send, entangled with some system he keeps.

All in all, we have a non-interactive reversed XOT protocol, implementing XOT from a party who only needs to detect quantum states, to a party who only needs to send states. The receiver of the quantum states does not need to test a fraction of the received states. The cheating probabilities for the two parties are the same as in the unreversed version of the protocol.

\subsection{Original and reversed protocols in terms of a shared entangled state}

Instead of preparing and sending one of the states $\ket{\phi_{x_0 x_1}}$ as in the original protocol, Alice could prepare the state
\begin{eqnarray}
\label{eq:entstate1}
\ket{\Psi_{\rm ent}}_{AB} &=& \frac{1}{2}\left(\ket{a}_A\ket{\phi_{00}}_B + \ket{b}_A\ket{\phi_{01}}_B\right.\nonumber\\
&&+ \left.\ket{c}_A\ket{\phi_{11}}_B + \ket{d}_A\ket{\phi_{10}}_B\right),
\end{eqnarray}
where $\ket a_A, \ket b_A, \ket c_A, \ket d_A$ is an orthonormal basis for a system she keeps on her side. She sends the $B$ system to Bob. If Alice measures the $A$ system in the $\ket a _A, \ket b_A, \ket c_A, \ket d_A$ basis, then this prepares one of the states $\ket{\phi_{x_0x_1}}_B$ on Bob's side. From both Bob's and Alice's viewpoints, this is equivalent to the original protocol, and their cheating probabilities remain the same. Using the definitions of $\ket{\phi_{x_0 x_1}}$ in Eq. \eqref{Honest Alices States}, the entangled state in Eq. \eqref{eq:entstate1} can also be written
\begin{eqnarray}
\label{eq:entstate2}
\ket{\Psi_{\rm ent}}_{AB} &=&\frac{1}{\sqrt 3}(\ket 0_A\ket 0_B+\ket 1_A\ket 1_B+\ket 2_A\ket 2_B),\nonumber
\end{eqnarray}
where
\begin{eqnarray}
\ket 0_A&=&\frac{1}{2}(\ket a_A+\ket b_A +\ket c_A + \ket d_A),\nonumber\\
\ket 1_A&=&\frac{1}{2}(\ket a_A-\ket b_A -\ket c_A + \ket d_A),\nonumber\\
\ket 2_A&=&\frac{1}{2}(\ket a_A+\ket b_A -\ket c_A - \ket d_A),\nonumber\\
\ket 3_A&=&\frac{1}{2}(\ket a_A-\ket b_A +\ket c_A - \ket d_A),
\label{eq:Akets}
\end{eqnarray}
are orthonormal states, and we have defined a fourth basis ket $\ket 3_A$. Both Alice's and Bob's state spaces for the state $\ket{\Psi_{\rm ent}}_{AB}$ are three-dimensional; its Schmidt number is 3. Alice's measurement in the $\ket a_A, \ket b_A, \ket c_A, \ket d_A$ basis can be understood as a realisation, with a Neumark extension using the auxiliary basis state $\ket 3_A$, of her generalised quantum measurement in the ``reversed" protocol, with measurement operators given in Eq. \eqref{eq:Alicereversemeas}.

If, instead, Bob prepares the state $\ket{\Psi_{\rm ent}}_{AB}$, sends the $A$ system to Alice, and measures his $B$ system using the measurement he makes in the original protocol, then this prepares one of the states $\ket{\phi_{x_i=b}}$ on Alice's side. This is equivalent to the reversed protocol. That is, starting from the entangled state $\ket{\Psi_{\rm ent}}_{AB}$, either the original or the reversed protocol can be implemented. In both cases, Alice makes the measurement she would make in the reversed protocol, and Bob makes the measurement he would make in the original protocol. What determines whether the procedure is equivalent to the original or reversed protocol is who prepares the state $\ket{\Psi_{\rm ent}}_{AB}$. This matters, because Alice and Bob are not guaranteed to follow the protocol and could prepare some other state if they were dishonest.

In order to illustrate the generality of this procedure, let us ``reverse" the protocol in~\cite{Amiri21}. Here, Alice sends one of the two-qubit states $\ket 0\ket 0, \ket +\ket +, \ket 1\ket 1, \ket -\ket -$ to Bob, encoding her bit values $00, 01, 11$, and $10$, respectively, with $\ket\pm=(\ket 0\pm\ket 1)/\sqrt 2$. Since each of Alice's four states are symmetric under exchange of the two qubits, the state space is actually three-dimensional. Bob measures one qubit in the $\ket 0, \ket 1$ basis and the other one in the $\ket +, \ket -$ basis, which allows him to rule out two of the four possible states, so that he can infer the value of either Alice's first bit or her second bit (never their XOR).
We will here use an equivalent set of four states with the same pairwise overlaps,
\begin{equation}
\ket{\phi'_{x_0x_1}}=\frac{1}{\sqrt 2}\ket 0+(-1)^{x_1}\frac{1}{2}\ket 1+(-1)^{x_0}\frac{1}{2}\ket 2,
\end{equation}
making it immediately clear that the state space is three-dimensional. Alice could now instead prepare the state
\begin{eqnarray}
\ket{\Phi'_{\rm ent}} &=& \frac{1}{2}(\ket a_A\ket{\phi'_{00}}_B+ \ket b_A\ket{\phi'_{01}}_B\nonumber\\
&&+ \ket c_A\ket{\phi'_{11}}_B+ \ket d_A\ket{\phi'_{10}}_B)\\
&=&\frac{1}{\sqrt 2}\ket 0_A\ket 0_B + \frac{1}{2}\ket 1_A\ket 1_B+ \frac{1}{2}\ket 2_A\ket 2_B,\nonumber
\end{eqnarray}
with the same definition of the states $\ket 0_A, \ket 1_A, \ket 2_A$ as in Eq. \eqref{eq:Akets}, and send the $B$ system to Bob. If Alice measures the $A$ system in the $\ket a_A, \ket b_A, \ket c_A, \ket d_A$ basis -- or makes the equivalent four-outcome generalised measurement in the three-dimensional space spanned by the states $\ket 0_A, \ket 1_A, \ket 2_A$ -- then this prepares one of the states $\ket{\phi'_{x_0x_1}}$ on Bob's side. This is then equivalent to Alice's actions in the protocol in~\cite{Amiri21}. Preparing the above entangled state is also how a dishonest Alice would cheat in~\cite{Amiri21}. She can then always revert to effectively sending Bob one of the states she should have sent him, if Bob decides to test the state she has sent. If she does go ahead with cheating, she measures the $A$ system in a way that optimally lets her deduce which bit value ($x_0$ or $x_1$) Bob has obtained, in which case she can learn this with probability $3/4$. If Bob does not test the states Alice sends, one can show that she can in fact cheat with probability 1 (and that it does not help if Bob randomly chooses which qubit he measures in what basis). If Bob is dishonest, he can determine which one of the four states $\ket{\phi'_{x_0x_1}}$ Alice has sent him with probability $\approx 0.729$.

If, instead, Bob would prepare the entangled state and send the $A$ system to Alice, we obtain a reversed version of the protocol in~\cite{Amiri21}. 
One can show that the measurement an honest Bob makes prepares one of the four states 
\begin{eqnarray}
\ket{\phi'_{x_0=0}}&=&\frac{1}{\sqrt 2}(\ket a_A+\ket b_A),\nonumber\\
\ket{\phi'_{x_1=0}}&=&\frac{1}{\sqrt 2}(\ket a_A+\ket d_A),\nonumber\\
\ket{\phi'_{x_0=1}}&=&\frac{1}{\sqrt 2}(\ket c_A+\ket d_A), \nonumber\\
\ket{\phi'_{x_1=1}}&=&\frac{1}{\sqrt 2}(\ket b_A+\ket c_A)
\end{eqnarray}
on Alice's side. If Bob is dishonest, he can send Alice some other state(s) in the three-dimensional state space spanned by these states; the state $(\ket a_A -\ket b_A-\ket c_A+\ket d_A)/2$ is orthonormal to all of the above four states.
It is necessary to re-analyze what the cheating probabilities are in the reversed 1-2 OT protocol, since they may change when a protocol is reversed. Bob's cheating probability could increase, since more cheating strategies are available to him, while Alice's cheating probability could decrease. In this case, it can be shown that Alice's cheating probability is $3/4$, the same as in the unreversed protocol when Bob was testing the states she sends. Bob's cheating probability increases to $3/4$ if Alice does not test the states he sends, and is equal to 0.729 if Alice does test the states he sends (the same as Bob's cheating probability in the unreversed protocol).

\begin{figure*}[!htb]
	\includegraphics[width=0.8\textwidth]{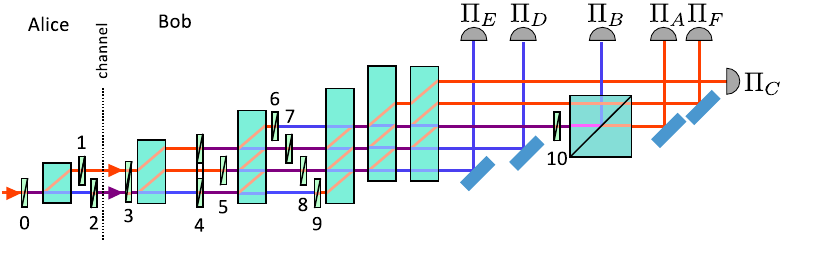}
	\caption{Experimental setup for the XOT protocol when Bob is honest. Green boxes labeled with black numbers represent half-wave plates (HWP). Large semi-transparent cyan boxes represent beam displacers. Next to HWP10, there is a polarizing beam-splitter. Note that HWP4 is ring-shaped and polarization of the central beam is not affected. Detectors are labeled according to the corresponding POVM operators. Settings for Alice's half-wave plates are listed in Tab.~\ref{EXP-tab-direct-honest-angles} for when she is honest and in Tab.~\ref{EXP-tab-direct-alice-angles} for when she is cheating. Settings for Bob's half-wave plates are HWP3$=0^\circ$, HWP4=HWP5=HWP7=HWP8=HWP10$=22.5^\circ$, HWP6=HWP9$=45^\circ$. Beams marked in red have horizontal linear polarization, beams marked in blue have vertical polarization. Purple indicates general polarization states.}
	\label{EXP-scheme-direct-honest}
\end{figure*}

\section{Experiment} \label{sec:Experiment}

\subsection{Both parties honest}

In the case where both parties are honest, Alice prepares one of the qutrit states given in Eq. \eqref{Honest Alices States} and sends it to Bob. In our experimental implementation, these states are encoded into spatial and polarization degrees of freedom of a single photon using the setup depicted in Alice's part of Fig.~\ref{EXP-scheme-direct-honest}, which consists of half-wave plates and a calcite beam displacer (it shifts the horizontally polarized beam). The basis state $|0\rangle$ is represented by the horizontally polarized mode in the upper output, $|1\rangle$ by the horizontally polarized mode in the lower output, and $|2\rangle$ by the vertically polarized mode in the lower output. The photons incoming from a single-photon source have a horizontal linear polarization. The settings of the angles of the wave-plate axes corresponding to all of Alice's states are listed in Table ~\ref{EXP-tab-direct-honest-angles}. 

\begin{table}[!htb]
	\begin{tabular}{l|r|r|r|r}
		& \multicolumn{1}{|c}{$| \phi_{00} \rangle$} & \multicolumn{1}{|c}{$| \phi_{01} \rangle$} & \multicolumn{1}{|c}{$| \phi_{10} \rangle$} & \multicolumn{1}{|c}{$| \phi_{11} \rangle$} \\
		\hline
		HWP0 & $-27.37^{\circ}$ & $-27.37^{\circ}$ & $27.37^{\circ}$ & $27.37^{\circ}$  \\
		HWP2 & $-25.50^{\circ}$ & $25.50^{\circ}$ &  $25.50^{\circ}$ & $-25.50^{\circ}$  \\
		\hline
	\end{tabular}
	\caption{Wave-plate angles for Alice's state preparation if Alice is honest. The angle of HWP1 is always zero (it only compensates for path differences). These settings also hold for cheating Bob in the reversed protocol.}
	\label{EXP-tab-direct-honest-angles}
\end{table}

Bob's six measurement operators are defined in Table \ref{tab:BobMeOp}, for unambiguously eliminating pairs of states. This measurement can be implemented by a projective von Neumann measurement $\{| \xi_i \rangle\!\langle \xi_i |\}_{i=A}^F$ in an extended six-dimensional Hilbert space, where
\begin{align}
	| \xi_{A} \rangle &= \frac{1}{2} \left( |0\rangle + |2\rangle + |3\rangle - |5\rangle \right), \nonumber \\
	| \xi_{B} \rangle &= \frac{1}{2} \left( |0\rangle - |2\rangle + |3\rangle + |5\rangle \right), \nonumber \\
	| \xi_{C} \rangle &= \frac{1}{2} \left( |0\rangle + |1\rangle - |3\rangle + |4\rangle \right), \nonumber \\
	| \xi_{D} \rangle &= \frac{1}{2} \left( |0\rangle - |1\rangle - |3\rangle - |4\rangle \right), \nonumber \\
	| \xi_{E} \rangle &= \frac{1}{2} \left( |1\rangle + |2\rangle - |4\rangle + |5\rangle \right), \nonumber \\	
	| \xi_{F} \rangle &= \frac{1}{2} \left( |1\rangle - |2\rangle - |4\rangle - |5\rangle \right), 
	\label{EXP-eq-6d-projectors}
\end{align}
are orthogonal states with $|3\rangle, |4\rangle$, and $|5\rangle$ being the basis states in the additional dimensions represented by modes added on Bob's side.

\begin{figure*}[!htb]
	\centering
	\includegraphics[width=0.8\textwidth]{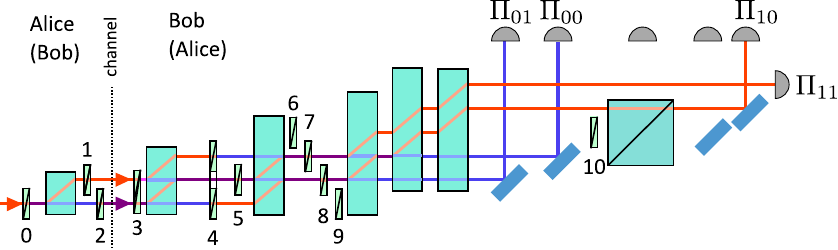}
	\caption{Experimental setup for the XOT protocol when Bob is cheating. The notation is the same as in Fig.~\ref{EXP-scheme-direct-honest}. The settings for the receiver's half-wave plates are HWP3=HWP7=HWP8=$22.5^\circ$, HWP4=$45^\circ$, HWP5$=90^\circ$. The same setup is used for the reversed protocol when Alice is honest. But in that case, Bob is the sender and Alice is the receiver (names in parentheses). The settings of the sender's half-wave plates for honest Alice in the unreversed protocol, or for cheating Bob in the reversed protocol, are listed in Tab.~\ref{EXP-tab-direct-honest-angles}, and for honest Bob in the reversed protocol in Tab.~\ref{EXP-tab-reversed-honest-angles}.}
	\label{EXP-scheme-direct-bob}
\end{figure*}

A unitary transformation between states $\{| \xi_i \rangle\}_{i=A}^F$ and the computational basis $\{| j \rangle\}_{j=0}^5$ can be realized by a symmetric beam-splitter network (consisting of six 50:50 beam splitters), which can be further translated into a setup consisting of half-wave plates and a beam displacer which combines spatial and polarization modes of light. 
The first beam displacer on Bob's side in Fig.~\ref{EXP-scheme-direct-honest} just transfers the incoming polarization and spatial modes into three separate paths.
The following half-wave plates -- ``double'' HWP4 and HWP5 -- turned by $22.5^\circ$ play the role of ``beam splitters'', mixing the original three modes with the additional three ``empty'' (vacuum) modes. Each wave-plate mixes two polarization modes. Behind the next beam displacer there are two half-wave plates turned by $45^\circ$, which just swap horizontal and vertical linear polarizations, and two half-wave plates turned by $22.5^\circ$ which represent the other two ``beam splitters''. The last ``beam splitter'' of the network is implemented by a half-wave plate turned by $22.5^\circ$ followed by a polarizing beam splitter in the right part of the figure.

To prevent the injection of higher-dimensional states into the Bob's apparatus (so that Alice only has access to the subspace spanned by her legitimate states), there should be a linear polarizer placed in the upper input port. However, in our proof-of-principle experiment we have omitted it in order to simplify the setup. 

The measurement results are shown in Table ~\ref{EXP-tab-direct-honest} of Appendix ~\ref{appendix:ex data}. This gives the absolute numbers of detector counts, the corresponding relative frequencies, and theoretical probabilities for comparison. The digits in parenthesis represent one standard deviation at the final decimal place. The states in Eq. \eqref{Honest Alices States} were being prepared with equal probabilities. 
The average error rate caused by experimental imperfections was 0.01249(8).
It was calculated as $\frac{\sum_{i,j\in {\cal E}_i} C_{ij}}{\sum_{i,j} C_{ij}}$, where $i$ indexes input states, $j$ indexes measurement results, $C_{ij}$ are measured numbers of counts, and ${\cal E}_i$ denote the sets of erroneous outcomes (outcomes that should not occur).

\subsection{Alice cheating}

Bob is honest, so his measurement is the same as in the previous case. To guess which of the three bits Bob will obtain, Alice sends states $|0 \rangle, |1 \rangle,$ or $|2 \rangle$.  The corresponding angles of the wave-plates are listed in Table ~\ref{EXP-tab-direct-alice-angles}.

\begin{table}[!htb]
	\begin{tabular}{l|r|r|r}
		& \multicolumn{1}{|c}{$| 0 \rangle$} & \multicolumn{1}{|c}{$| 1 \rangle$} & \multicolumn{1}{|c}{$| 2 \rangle$} \\
		\hline
		HWP0 & $0^{\circ}$ & $45^{\circ}$ & $45^{\circ}$  \\
		HWP2 & $0^{\circ}$ & $45^{\circ}$ &  $0^{\circ}$  \\
		\hline
	\end{tabular}
	\caption{Angles for wave plates, for Alice's state preparation if Alice is cheating. The angle of HWP1 is always zero.}
	\label{EXP-tab-direct-alice-angles}
\end{table}

The measurement results are shown in Table ~\ref{EXP-tab-direct-alice} of Appendix ~\ref{appendix:ex data}. Alice's states were being prepared with equal probabilities. 
Her average probability of correctly guessing which one of the three bits Bob obtained (i.e. his value of $b$), estimated from the experiment, was 0.4999(3).
It was calculated as $\frac{\sum_{i,j\in {\cal C}_i} C_{ij}}{\sum_{i,j} C_{ij}}$, where ${\cal C}_i$ denote the sets of correct guesses.
The theoretical prediction is $1/2$.

\subsection{Bob cheating}

Alice is honest, so she sends her states exactly as in the described case above, when both parties were honest. To guess all three bits (equivalently, any two bits), Bob applies the square-root measurement consisting of four positive operator-valued measure (POVM) elements, which are actually the same as that expressed in Eq. \eqref{eq:Alicereversemeas}. This POVM can be implemented by projectors $\{| \xi_i \rangle\!\langle \xi_i |\}_{i=00}^{11}$ in a four-dimensional Hilbert space spanned by $|0\rangle, |1\rangle, |2\rangle, |3\rangle $, where  
\begin{align}
	| \xi_{00} \rangle &= \frac{1}{2} \left( |0\rangle +|1\rangle +|2\rangle +|3\rangle \right), \nonumber \\
	| \xi_{01} \rangle &= \frac{1}{2} \left( |0\rangle -|1\rangle +|2\rangle -|3\rangle \right), \nonumber \\
	| \xi_{10} \rangle &= \frac{1}{2} \left( |0\rangle +|1\rangle -|2\rangle -|3\rangle \right), \nonumber \\	
	| \xi_{11} \rangle &= \frac{1}{2} \left( |0\rangle -|1\rangle -|2\rangle +|3\rangle \right),
	\label{EXP-eq-4d-projectors}
\end{align}
are orthogonal states.
The implementation of this projective measurement is shown in Fig.~\ref{EXP-scheme-direct-bob}. The angles of the wave plates are listed in the figure caption.

The measurement results are shown in Table ~\ref{EXP-tab-direct-bob} of Appendix ~\ref{appendix:ex data}. Alice's states were being prepared with equal probabilities. 
Bob's average probability of guessing all bits, estimated from the experiment, was 0.7431(3).
The theoretical value is $3/4$.

\begin{figure*}[!htb]
	\centering
	\includegraphics[width=0.8\textwidth]{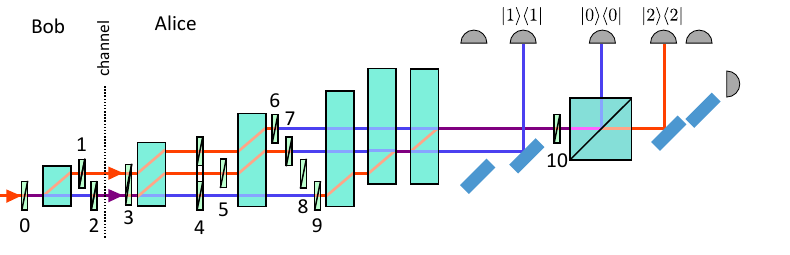}
	\caption{Experimental setup for the reversed XOT protocol when Alice is cheating. The notation is the same as in Fig.~\ref{EXP-scheme-direct-honest}. The settings of Alice's half-wave plates are $\text{HWP3}=\text{HWP4}=\text{HWP10}=0^\circ$, $\text{HWP5}=90^\circ$, $\text{HWP6} =$ $ \text{HWP7}$ $ = \text{HWP9}=45^\circ$.}
	\label{EXP-scheme-reversed-alice}
\end{figure*}

\subsection{Reversed protocol -- Both parties honest}

In the reversed protocol, Bob prepares and sends one of the six non-orthogonal qutrit states defined in ~\eqref{eq:Bobreversedstates}. These states can be prepared in a similar way as Alice's states were being prepared in the unreversed protocol. The corresponding angles for the wave-plates are listed in Table ~\ref{EXP-tab-reversed-honest-angles}.

In this case, Alice is the receiver. To learn the bit values she performs a POVM measurement, the components of which are defined in ~\eqref{eq:Alicereversemeas}. We  already know how to implement this measurement, because it is exactly the same as the measurement for cheating Bob in the unreversed protocol. So the corresponding higher-dimensional projective measurement consists of the projectors onto the states (\ref{EXP-eq-4d-projectors}). 
Therefore, the setup for the reversed protocol in the case when both parties are honest is actually the same as the setup for the unreversed protocol when Bob is cheating (see Fig.~\ref{EXP-scheme-direct-bob}) except that the roles of Alice and Bob are interchanged.

\begin{table}[!htb]
	\begin{tabular}{l|r|r|r|r|r|r}
		& $| \phi_{x_0=0} \rangle$ & $| \phi_{x_0=1} \rangle$ & $| \phi_{x_1=0} \rangle$ & $| \phi_{x_1=1} \rangle$ & $| \phi_{x_2=0} \rangle$ & $| \phi_{x_2=1} \rangle$ \\
		\hline
		HWP0 & $-22.5^{\circ}$ & $22.5^{\circ}$ & $22.5^{\circ}$ & $-22.5^{\circ}$ & $45.0^{\circ}$ & $45.0^{\circ}$ \\
		HWP2 & $0.0^{\circ}$ & $0.0^{\circ}$ & $45.0^{\circ}$ & $45.0^{\circ}$ & $-22.5^{\circ}$ & $22.5^{\circ}$ \\
		\hline
	\end{tabular}
	\caption{Reversed protocol. The wave plate angles for Bob's state preparation, if Bob is honest. The angle of HWP1 is always zero. $x_2 = x_0 \oplus x_1$.}
	\label{EXP-tab-reversed-honest-angles}
\end{table}

The measurement results are shown in Table ~\ref{EXP-tab-reversed-honest} of Appendix ~\ref{appendix:ex data}. Bob's states were being prepared with equal probabilities. 
The average error rate caused by experimental imperfections was 0.00428(4).

\subsection{Reversed protocol -- Alice cheating}

Bob honestly prepares his quantum states but cheating Alice wants to know which bit Bob has actually learned. In this case, Alice's optimal strategy is to use the measurement defined in Eq.~\eqref{eq:AliceCheatREv}. These POVM operators are actually statistical mixtures of the projectors onto the basis states $|0 \rangle, |1 \rangle,$ and $|2 \rangle$. This means that Alice can make a projective measurement followed by classical post-processing. E.g., if she obtains the result corresponding to $| 0 \rangle\!\langle 0 |$, she knows that Bob has either the value of bit $x_0$ or the value of bit $x_1$, each with 50\% probability.
The scheme of the setup implementing Alice's measurement if she is cheating is shown in Fig.~\ref{EXP-scheme-reversed-alice}. The angles of the wave plates are listed in the figure caption.

The measurement results are shown in Table ~\ref{EXP-tab-reversed-alice} of Appendix ~\ref{appendix:ex data}. Bob's states were being prepared with equal probabilities. 
The average probability of Alice guessing Bob's $b$, estimated from the experiment, was 0.4992(2).
The theoretical value is $1/2$.

\subsection{Reversed protocol -- Bob cheating}

In this case, Alice behaves honestly but cheating Bob wants to obtain the values of both $x_0$ and $x_1$ (and their XOR). To estimate these values, Bob uses a set of four ``fake" states, which are equivalent to the states in Eq. \eqref{Honest Alices States}. Clearly, the experimental setup, as well as the state preparation and measurement, are the same as that for the unreversed protocol with cheating Bob, see Fig.~\ref{EXP-scheme-direct-bob}. Therefore, it was not necessary to repeat the measurement because the results had already been obtained. They are shown in Table ~\ref{EXP-tab-direct-bob} of Appendix ~\ref{appendix:ex data}.
The average probability of Bob guessing all bits, estimated from the experiment, was 0.7431(3).
The theoretical value is $3/4$.

\begin{figure*}[htb]
	\centering
	\includegraphics[width=0.8\textwidth]{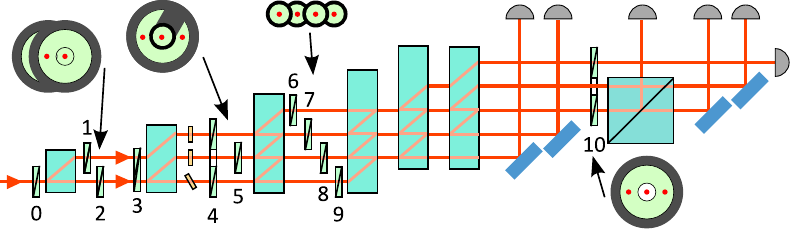}
	\caption{Detailed scheme of the experimental setup. Green boxes labeled with black numbers represent half-wave plates (HWP). Small orange rectangles are glass plates which serve for phase compensation. Large semi-transparent cyan boxes represent beam displacers. Next to HWP10, there is a polarizing beam-splitter. Note that HWP1, HWP2, HWP4, and HWP10 are ring-shaped and polarization of the central beam is not affected. Insets show the actual arrangement of the half-wave plates.}
	\label{EXP-scheme-full}
\end{figure*}

\subsection{Technical description of the setup}

We used a heralded single-photon source. Heralding was based on the detection of one photon from a time-correlated pair. Photon pairs were generated using type-II spontaneous parametric down-conversion in a periodically poled KTP crystal, and their wavelength was $810$~nm. The photons were entering the setup with linear horizontal polarization. Single-photon detection was implemented as coincidence measurements with the trigger signal heralding photon creation. The used coincidence window was $2.5$~ns. If more detectors clicked together with the trigger signal, only one result was randomly selected and counted. However, such situations occurred at most once in 2000 measurements.

Each qutrit used in the protocol is represented by a single photon which can occur in a superposition of three optical modes.
The protocol requires an interferometric network that allows coupling these modes with each other, and with a vacuum. Our implementation, depicted in Fig.~\ref{EXP-scheme-full}, is based on calcite beam displacers, which allow us to construct passively stable interferometers ~\cite{Starek18}. We used spatial and polarization degrees of freedom to encode the qutrits. This encoding enables us to realize a tunable beam splitter simply with a half-wave plate which couples horizontally and vertically polarized optical modes. The calcite beam displacers spatially separate horizontally and vertically polarized components into two parallel beams with 6~mm lateral distance. 

Although there are multiple optical paths, only four interferometric phases are important for the tested protocols. 
To adjust these phases, we used auxiliary wave plate settings such that the output optical signal was sensitive to the optical phase. The first relative optical phase was set by tilting the second beam displacer using a piezo-stack actuator attached to the prism turn-table. Then we adjusted the second phase by tilting the third beam displacer. The third phase was set by tilting the glass plate in the bottom arm. Finally, we set the last optical phase by tilting the fourth beam displacer. The phases have to be set in this order due to the sharing of optical paths.

Using strong laser light, we characterized the phase stability of the largest interferometer formed by the outer-most optical paths between the first and the fourth beam displacers, which merge at the sixth displacer. We set the optical phase roughly to $\pi/2$, covered the setup with a cardboard box, and monitored the output intensity for one hour. The observed drift speed was 0.5~deg/min. The amplitude of fast phase fluctuation was roughly 5~deg peak-to-peak. 

There were several sources of experimental errors. The most significant of them was the unequal fiber-coupling efficiencies at the output of the interferometric network, spanning from 0.75 to 0.85. Furthermore, the efficiencies of the used single-photon detectors were also unequal. The largest relative difference was 0.12. We compensated for these inequalities using detection electronics. The inaccurate retardance of half-wave plates causes the discrepancy between the expected and actual coupling ratio for a given angular position. We tried to compensate for this imperfection by slight modifications of the angular positions. Also, we used wave plates to exchange the polarization modes. The imperfect retardance limits the ability to turn the horizontal into vertical polarization. It consequently causes undesired losses and residual coupling in our experiments.
Furthermore, the slight variation in the length of the beam displacers causes imperfect overlap of optical beams, reducing the interferometric visibility. The worst visibility we observed was 0.85. Fortunately, coupling the beams into single-mode optical fibers serves as spatial filtering and restores interferometric visibility. The worst observed visibility of fiber-coupled signal was 0.99. We also observed that different optical paths suffered from slightly unequal optical losses (the largest difference was 0.02), but we did not directly compensate for this imperfection.

The counts $C_{ij}$ were accumulated during 10-second measurements for each input state.
Relative frequencies were calculated as $f_{ij}=\frac{C_{ij}}{\sum_{j} C_{ij}}$, where $i$ is indexing input states and $j$ is indexing measurement results. Shown errors of relative frequencies were determined using the standard law of error propagation under the assumption that the detection events obey the Poissonian distribution and thus the standard deviations of $C_{ij}$ can be estimated as $\sqrt{C_{ij}}$. 
The key part of the experiment is a stable optical realization of the
required POVM measurements.

\section{Conclusions}

We have analysed and realised protocols for quantum XOR oblivious transfer. The protocols are non-interactive, do not require entanglement, and make use of pure symmetric states. We presented particular optimal quantum protocols, showing that they outperform classical XOR oblivious transfer protocols, and obtained cheating probabilities for sender and receiver for general non-interactive symmetric-state protocols. The cheating probabilities for the protocols are the same as for a previous protocol~\cite{Kundu22}, which is interactive and requires entanglement. 
Non-interactive protocols which do not require entanglement are, however, simpler to implement. In our protocol, Bob obtains Alice's first bit, her second bit, or their XOR at random. Thus, we introduced the concept of semi-random XOT protocols, analogous to the definition of semi-random 1-2 OT protocols given in~\cite{Amiri21}, proving that a semi-random XOT protocol can be changed into a standard XOT protocol and vice versa by adding classical post-processing, keeping the cheating probabilities the same.

One can argue that the ``quantum advantage" for the presented quantum XOT oblivious transfer protocol is greater than that of the quantum 1-out-of-2 oblivious transfer protocol in~\cite{Amiri21}. In addition, the cheating probabilities for the protocol in~\cite{Amiri21} are average cheating probabilities for many rounds of oblivious transfer. A sender can cheat with probability 1 in any single round, with negligible probability of being caught, as long as the average cheating probability obeys the bound. For the XOR oblivious transfer protocol we presented, the sender's cheating probability in every single round is bounded by 1/2.

We also introduced the concept of ``reversing" a protocol, which means that the sender of the quantum state instead becomes a receiver of quantum states, and vice versa, while keeping their roles in the XOT protocol the same. This is useful if one party only has the ability to prepare and send quantum states, while the other party can only measure them. 
This is frequently the case in quantum communications systems. The ``original" and ``reversed" protocols can be connected by viewing them in terms of a shared entangled state.
Because the two parties do not trust each other in oblivious transfer, or in multiparty computation more generally, unlike for quantum key distribution, cheating probabilities can be different depending on who prepares the entangled state. For our XOT protocol, however, cheating probabilities are the same in the unreversed protocol and in its reversed version.

We optically realised both the unreversed and the reversed version of our optimal non-interactive quantum XOT protocol, including Alice's and Bob's optimal cheating strategies. 
The experiment involved the implementation of the generalised quantum measurements made by the receiver Bob in the unreversed protocol, and the sender Alice in the reversed protocol. This was achieved through extending the Hilbert space using an auxiliary basis state, which was coupled to the other basis states using a particular unitary transform. The experimental work involved aligning and stabilizing several concatenated Mach-Zehnder interferometers, which was not trivial. Generalised measurements are still quite rare in quantum communication and quantum cryptographic protocols, which mostly use standard projective quantum measurements.
The achieved experimental data match our theoretical results very well, thus demonstrating the feasibility of both protocols.

\section*{Acknowledgements}

This work was supported by the UK Engineering and Physical Sciences Research Council (EPSRC) under Grants No. EP/T001011/1 and EP/R513386/1. R.S., N.H., and M.D. acknowledge support by Palack\'{y} University under Grants No. IGA-PrF-2021-006 and IGA-PrF-2022-005.

\bibliography{Noninteractive_XOT}

\appendix 

\section{Quantum XOT with symmetric states: details of derivations}

\subsection{Conditions involving pairwise overlaps}
\label{app:FGcond}

We first give a number of useful relations. For a set of symmetric pure states, the pairwise overlaps obey 
\begin{eqnarray}
\braket{\psi_{01} | \psi_{00}} &=& \braket{\psi_{11} | \psi_{01}} = \braket{\psi_{10} | \psi_{11}} = \braket{\psi_{00} | \psi_{10}}=F, \nonumber\\
\braket{\psi_{00} | \psi_{11}} &=& \braket{\psi_{01} | \psi_{10}}=G.
\end{eqnarray}
Since $\ket{\psi_{11}}=U^2\ket{\psi_{00}}$ and the eigenvalues of $U$ are the 4$^\text{th}$ roots of unity, $G$ is always real, but $F$ is in general complex. 
We denote an honest Bob's measurement operators by $ \Pi_{0*}, \Pi_{1*}, \Pi_{*0}, \Pi_{*1}, \Pi_{\text{XOR}=0}, \Pi_{\text{XOR}=1}$ (using different indices than in Section \ref{sec:XOT protocol}, in order to distinguish these more general measurement operators from the specific ones in Section \ref{sec:XOT protocol}). Bob should obtain either the first or second bit, or their XOR, each with probability $1/3$. The probability of obtaining outcome $m$ is 
\[p_m=\bra{\psi_{jk}}\Pi_{m}\ket{\psi_{jk}},\]
 for $m \in \{0*, 1*, *0, *1, \text{XOR}=0, \text{XOR} =1\}$. This probability should be equal to $1/3$ when an outcome is possible and otherwise be equal to $0$. Moreover, it holds that
\begin{eqnarray}
\label{eq:FGcond}
\bra{\psi_{01}}\Pi_{0*}\ket{\psi_{00}} &=& 
\bra{\psi_{10}}\Pi_{1*}\ket{\psi_{11}} =\\ 
\bra{\psi_{00}}\Pi_{*0}\ket{\psi_{10}} &=&
\bra{\psi_{11}}\Pi_{*1}\ket{\psi_{01}} = F,\nonumber\\ 
\bra{\psi_{00}}\Pi_{\text{XOR}=0}\ket{\psi_{11}} &=& 
\bra{\psi_{01}}\Pi_{\text{XOR}=1}\ket{\psi_{10}} =G.\nonumber
\end{eqnarray}
The above relations can be obtained by writing $\Pi_{0*}=\sum_k\lambda_k\ket{\lambda_k}\bra{\lambda_k}$ in terms of its eigenstates and eigenvalues, and similar for the other measurement operators. It then holds, for example, that
\begin{equation}
0=\bra{\psi_{10}}\Pi_{0*}\ket{\psi_{10}} = 
\sum_k\lambda_k|\braket{\psi_{10}|\lambda_k}|^2,
\end{equation}
meaning that $\braket{\psi_{10}|\lambda_k}=0~ \forall~ k$.
Using this and other analogous conditions, we have
\begin{equation}
F=\braket{\psi_{01}|\psi_{00}}=
\bra{\psi_{01}}\sum_m\Pi_m\ket{\psi_{00}} =
\bra{\psi_{01}}\Pi_{0*}\ket{\psi_{00}}.
\end{equation}
The other conditions in Eq. \eqref{eq:FGcond} can be obtained analogously.
Furthermore, it has to hold that $|F|\le 1/3$ and $|G|\le 1/3$. This is necessary for the states to be distinguishable enough, so that an honest Bob can learn either $x_0$, $x_1$, or $x_0\oplus x_1$ correctly. To show this, define a vector $\bf X$ with elements $x_k=\sqrt{\lambda_k}\braket{\psi_{01}|\lambda_k}$ and a vector $\bf Y$ with elements $y_k=\sqrt{\lambda_k}\braket{\psi_{00}|\lambda_k}$. Then $|{\bf X}|^2=|{\bf Y}|^2=1/3$, and it holds that
\begin{eqnarray}
|F|^2&=&|\sum_k\lambda_k\braket{\psi_{01}|\lambda_k}\braket{\lambda_k|\psi_{00}}|^2=|\sum_k x_k y_k^*|^2 \nonumber\\
&\le& |{\bf X}|^2|{\bf Y}|^2=\frac{1}{9}.
\end{eqnarray}
$|G|\le 1/3$ can be proven analogously.

\subsection{Alice's cheating probability when Bob is testing her states} 
\label{appendix:Bound cheating Alice testing Bob}

In order to always be able to pass Bob's tests, a dishonest Alice must use an equal superposition of the states she is supposed to send, entangled with a system she keeps on her side. That is a state of the form
\begin{eqnarray} \label{appendix:AliceCheatState}
\ket{\Psi_{\text{cheat}}} &=&  a\ket{0}_A \otimes \ket{\psi_{00}} +   b\ket{1}_A \otimes \ket{\psi_{01}} \nonumber \\
&&+   c\ket{2}_A \otimes \ket{\psi_{11}} +   d\ket{3}_A \otimes \ket{\psi_{10}},
\end{eqnarray}
where $ \ket{0}_A, \ket{1}_A, \ket{2}_A, \ket{3}_A$ is an orthonormal basis for her kept system and $|a|^2+|b|^2+|c|^2+|d|^2=1$. 
If Alice measures her system in this basis, she can always pass any tests Bob might conduct. It will appear to Bob as if Alice is sending one of the four states she should send if she is honest.

We will set $a=b=c=d=1/2$, which we conjecture is actually optimal for Alice. Any cheating strategy for Alice will nevertheless give a lower bound on her cheating probability.
Honest Bob performs a measurement with measurement operators $\Pi_{0*}$, $\Pi_{1*}$, $\Pi_{*0}$, $\Pi_{*1}$, $\Pi_{\text{XOR}=0}$, $\Pi_{\text{XOR}=1}$ on the system he receives.
Using the conditions in Eq. \eqref{eq:FGcond}, we can express the equiprobable states Alice holds, conditioned on Bob's $b$, as
\begin{align}
\mu_A^{b=0} &=\dfrac{1}{4}
\begin{pmatrix}
1 & 3F & 0 & 0 \\
3F^* & 1 & 0 & 0 \\
0 & 0 & 1 & 3F \\
0 & 0 & 3F^* & 1
\end{pmatrix} , \nonumber \\
\mu_A^{b=1} &= \dfrac{1}{4}
\begin{pmatrix}
1 & 0 & 0 & 3F^* \\
0 & 1 & 3F & 0 \\
0 & 3F^*  & 1 & 0 \\
3F & 0 & 0 & 1
\end{pmatrix} , \nonumber \\
\mu_A^{b=2} &= \dfrac{1}{4}
\begin{pmatrix}
1 & 0 & 3G & 0 \\
0 & 1 & 0 & 3G \\
3G & 0  & 1 & 0 \\
0 & 3G  & 0 & 1
\end{pmatrix} ,
\label{eq:AliceBobtest}
\end{align}
corresponding to Bob obtaining $x_0$, $x_1$, or $x_2 = x_0\oplus x_1$. 
Here,
\begin{eqnarray*}
\mu_A^{b=0}&&= \\
&& \frac{{\rm Tr}_B[(\Pi_{0*}+\Pi_{1*})^{1/2} \ket{\Psi_{\rm cheat}}\bra{\Psi_{\rm cheat}}(\Pi_{0*}+\Pi_{1*})^{1/2}]}{p_{0*}+p_{1*}},
\end{eqnarray*}
where 
\[p_{0*}+p_{1*}={\rm Tr}[\ket{\Psi_{\rm cheat}}\bra{\Psi_{\rm cheat}}(\Pi_{0*}+\Pi_{1*})]=\frac{1}{3},\]
and analogously for $b=1$ and $b=2$.
These states are mirror-symmetric, meaning that the unitary transformation that takes $\ket{0} \rightarrow \ket{3}$, $\ket{3} \rightarrow \ket{2}$, $\ket{2} \rightarrow \ket{1}$, and $\ket{1} \rightarrow \ket{0}$, takes $\mu_A^{b=0}$ to $\mu_A^{b=1}$ and vice versa, and keeps $\mu_A^{b=2}$ unchanged. The minimum-error measurement is known for some sets of mirror-symmetric states ~\cite{Andersson02, Chou04}, but this one is not one of them. Alice's minimum-error measurement can nevertheless be found by making a basis transform using a unitary transform $U$ proportional to a $4\times 4$ Hadamard-Walsh matrix,
\begin{equation}
    U=\frac{1}{2}\begin{pmatrix}
    1&1&1&1\\
    1&-1&1&-1\\
    1&1&-1&-1&\\
    1&-1&-1&1
    \end{pmatrix}.
\end{equation}
If we interpret the four basis states as two-qubit states so that $\ket{0} \equiv \ket{00}, \ket{1}\equiv \ket {01}, \ket{2} \equiv \ket {10}, \ket{3} \equiv \ket{11}$, this is the same as writing the density matrices in the $\ket{++}, \ket{+-}, \ket{-+}, \ket{- -}$ basis, where $\ket{\pm} = (\ket{0} \pm \ket{1})/\sqrt 2$. The density matrices in Eq. \eqref{eq:AliceBobtest} then become
\begin{align}
\mu_A^{b=0} &= \dfrac{1}{4}
\begin{pmatrix}
1+3\text{Re}\, F & -3i\text{Im}\, F & 0 & 0 \\
3i\text{Im}\, F & 1-3\text{Re}\, F & 0 & 0 \\
0 & 0 & 1+3\text{Re}\, F & -3i\text{Im}\, F \\
0 & 0 & 3i\text{Im}\, F& 1-3\text{Re}\, F
\end{pmatrix} , \nonumber \\
\mu_A^{b=1} &= \dfrac{1}{4}
\begin{pmatrix}
1+3\text{Re}\, F & 3i\text{Im}\, F & 0 & 0 \\
-3i\text{Im}\, F & 1-3\text{Re}\, F & 0 & 0 \\
0 & 0 & 1-3\text{Re}\, F & -3i\text{Im}\, F \\
0 & 0 & 3i\text{Im}\, F& 1+3\text{Re}\, F
\end{pmatrix} , \nonumber \\
\mu_A^{b=2} &= \dfrac{1}{4}
\begin{pmatrix}
1+3G & 0 & 0 & 0 \\
0 & 1+3G & 0 & 0 \\
0 & 0  & 1-3G & 0 \\
0 & 0  & 0 & 1-3G
\end{pmatrix} .
\label{eq:AliceBobtest+-}
\end{align}
All three density matrices are block-diagonal. This means that the minimum-error measurement can be performed by first projecting on the subspaces corresponding to each block, which is the same as measuring the first qubit in the $\ket{+}, \ket{-}$ basis. Depending on the outcome, one then distinguishes between the resulting three density matrices in that subspace. In each subspace, $\mu_A^{b=2}$ is proportional to an identity matrix. This means that no measurement will tell Alice anything more about the likelihood that her state was $\mu_A^{b=2}$, other than what she already knows. The optimal measurement in each subspace is then the measurement that optimally distinguishes between $\mu_A^{b=0}$ and $\mu_A^{b=1}$. Depending on the outcome, Alice's best guess might still be $b=2$.

To summarise, the measurement which distinguishes between $\mu_A^{b=0}, \mu_A^{b=1}$, and $\mu_A^{b=2}$ with minimum error, and therefore maximises Alice's cheating probability, is a projection on the states $\ket{+R}, \ket{+L}, \ket{-+}, \ket{- -}$, where $\ket R = (\ket 0 +i\ket 1)/\sqrt 2$ and $\ket L = (\ket 0 -i\ket 1)/\sqrt 2$. The probabilities for the different outcomes, conditioned on what density matrix Alice holds, are
\begin{eqnarray}
p(+,R|\mu_A^{b=0})=p(+,L|\mu_A^{b=1}) &=& \frac{1}{4}(1+3\text{Im}\, F),\nonumber\\
p(+,L|\mu_A^{b=0})=p(+,R|\mu_A^{b=1}) &=& \frac{1}{4}(1-3\text{Im}\, F),\nonumber\\
p(+,L|\mu_A^{b=2})=p(+,R|\mu_A^{b=2}) &=& \frac{1}{4}(1+3G),\nonumber\\
p(-,+|\mu_A^{b=0})=p(-,-|\mu_A^{b=1}) &=& \frac{1}{4}(1+3\text{Re}\, F),\nonumber\\
p(-,-|\mu_A^{b=0})=p(-,+|\mu_A^{b=1}) &=& \frac{1}{4}(1-3\text{Re}\, F),\nonumber\\
p(-,+|\mu_A^{b=2})=p(-,-|\mu_A^{b=2}) &=& \frac{1}{4}(1-3G).
\end{eqnarray}
Given one of the four outcomes, Alice chooses the most likely value of $b$. Her cheating probability is then bounded as
\begin{equation}
\label{appendix:AliceCheatGenBound}
    A_{OT} \ge\left\{\begin{array}{c} \frac{1}{3}+\frac{1}{2}|\text{Im}\, F|  +\frac{1}{2}\max(|\text{Re}\, F|, |G|), \text{ if } G\le 0, \\
    \\
     \frac{1}{3}+\frac{1}{2}|\text{Re}\, F|+\frac{1}{2}\max(|\text{Im}\, F|, |G|), \text{ if } G > 0,\end{array}
\right.
\end{equation}
as given in Eq. \eqref{eq:AliceCheatGenBound}.

\subsection{Alice's cheating probability when Bob is not testing her states} 
\label{appendix:Bound cheating Alice no testing Bob}

It is optimal for Alice to send Bob the pure state, within the subspace spanned by the states she is supposed to send him, for which Bob's probability to obtain either $b=0, b=1$, or $b=2$ is maximised. Alice's state can be written
\begin{equation}
\ket{\Psi_{\rm cheat}} = \alpha \ket{\psi_{00}} + \beta \ket{\psi_{01}} + \gamma \ket{\psi_{11}} + \delta \ket{\psi_{10}},
\end{equation}
where $\alpha, \beta, \gamma, \delta$ are complex coefficients, chosen so that the state is normalised.
Bob's probabilities to obtain $b=0, b=1$, and $b=2$ are
\begin{eqnarray}
p(b=0)&=&\bra{\Psi_{\rm cheat}}\Pi_{0*}+\Pi_{1*}\ket{\Psi_{\rm cheat}},\\
p(b=1)&=&\bra{\Psi_{\rm cheat}}\Pi_{*0}+\Pi_{*1}\ket{\Psi_{\rm cheat}},\nonumber\\
p(b=2)&=&\bra{\Psi_{\rm cheat}}\Pi_{\rm XOR =0}+\Pi_{\rm XOR=1}\ket{\Psi_{\rm cheat}},\nonumber
\end{eqnarray}
which, using the conditions in Eq. \eqref{eq:FGcond}, can be written
\begin{eqnarray} 
\label{eq:AlicecheatsP012}
p(b=0)&=&\frac{1}{3}(|\alpha|^2+|\beta|^2+|\gamma|^2+|\delta|^2)\nonumber\\
&&+(\alpha\beta^*+\gamma\delta^*)F+(\alpha^*\beta+\gamma^*\delta)F^*,\nonumber\\
p(b=1)&=&\frac{1}{3}(|\alpha|^2+|\beta|^2+|\gamma|^2+|\delta|^2)\nonumber\\
&&+(\alpha^*\delta+\beta\gamma^*)F+(\alpha\delta^*+\beta^*\gamma)F^*,\nonumber\\
p(b=2)&=&\frac{1}{3}(|\alpha|^2+|\beta|^2+|\gamma|^2+|\delta|^2)\nonumber\\
&&+(\alpha^*\gamma+\beta^*\delta+\alpha\gamma^*+\beta\delta^*)G.
\end{eqnarray}
Alice should choose $\alpha, \beta, \gamma, \delta$ so as to maximise one of these probabilities. Normalisation means that $p(b=0)+p(b=1)+p(b=2)=1$. 
 
Bob's probabilities in Eq. \eqref{eq:AlicecheatsP012} can also be written as
\begin{eqnarray}
\label{eq:Alicecheatnotestgen}
p(b=0)&=&(\frac{1}{3}-|F|)\left(|\alpha|^2+|\beta|^2+|\gamma|^2+|\delta|^2\right)\nonumber\\
&&+|F|\left(|\alpha e^{i\theta_F}+\beta|^2+|\gamma e^{i\theta_F}+\delta|^2\right),\nonumber\\
p(b=1)&=&(\frac{1}{3}-|F|)\left(|\alpha|^2+|\beta|^2+|\gamma|^2+|\delta|^2\right)\nonumber\\
&&+|F|\left(|\alpha e^{-i\theta_F}+\delta|^2+|\beta+\gamma e^{-i\theta_F}|^2\right),\nonumber\\
p(b=2)&=&(\frac{1}{3}-|G|)\left(|\alpha|^2+|\beta|^2+|\gamma|^2+|\delta|^2\right)\nonumber\\
&&+|G|\left(|\alpha\pm\gamma|^2+|\beta \pm\delta|^2\right),
\end{eqnarray}
where in the expression for $p(b=2)$, we have $+$ if $G>0$ and $-$ if $G<0$. From the above expressions, we see that if $|F|=|G|=1/3$, then Alice can cheat perfectly unless $F=\pm 1/3$ and $G=-1/3$, or $F=\pm i/3$ and  $G=1/3$. Unless one of these conditions hold, Alice can make $p(b=2)$ and either $p(b=0)$ or $p(b=1)$ equal to zero, while the remaining probability is equal to 1.
To make $p(b=1)=p(b=2)=0$ when $G=1/3$, for example, Alice chooses $\alpha=-\delta e^{i\theta_F}=\beta e^{i\theta_F}=-\gamma$. Then $p(b=0)=1$, unless it holds that $e^{2i\theta_F}=-1$, which is the case for $F=\pm i/3$.  
As we will show, for $F=\pm 1/3, G=-1/3$, or $F=\pm i/3, G=1/3$, Alice's cheating probability is equal to $1/2$ whether Bob tests the state she sends him or not. We have already seen that these choices of phases for $F$ and $G$, when $|F|=|G|=1/3$, also minimise Bob's cheating probability in Eq. \eqref{eq:BoundBob}. We have now found that they are the optimal -- or even the only sensible -- choices more generally whenever $|F|=|G|=1/3$, when Bob does not test Alice's state (since otherwise Alice can cheat with probability 1).

We will now derive Alice's cheating probabilities as a function of $F$ and $G$. Bob's probabilities $p(b=0)$, $p(b=1)$, and $p(b=2)$ can be written in bilinear form as
\begin{eqnarray}
p(b=i)= (\alpha^*, \beta^*, \gamma^*, \delta^*)M_i (\alpha, \beta, \gamma,\delta)^T,
\end{eqnarray}
where $i=0, 1, 2$, and the matrices $M_0, M_1, M_2$ are given by
\begin{eqnarray}
M_0&=& \begin{pmatrix}
1/3 & F^* & 0 & 0 \\
F & 1/3 & 0 & 0 \\
0 & 0 & 1/3 &F^* \\
0 & 0 & F& 1/3
\end{pmatrix} , \nonumber\\
M_1&=& \begin{pmatrix}
1/3 & 0 & 0 & F \\
0 & 1/3 & F^* & 0 \\
0 & F & 1/3 &0 \\
F^* & 0 & 0& 1/3
\end{pmatrix} , \nonumber\\
M_2&=& \begin{pmatrix}
1/3 & 0 & G & 0 \\
0 & 1/3 & 0 & G \\
G & 0 & 1/3 &0 \\
0 & G & 0& 1/3
\end{pmatrix}.
\end{eqnarray}
The normalisation condition is then written 
\begin{eqnarray}
\hskip-0.5cm(\alpha^*, \beta^*, \gamma^*, \delta^*)(M_0+M_1+M_2)(\alpha, \beta, \gamma,\delta)^T=1~~~
\label{eq:ellipsenormalisation}
\end{eqnarray}
and can be viewed as an ellipsoid in a four-dimensional complex space. The conditions
\begin{equation}
    p(b=0)=C_0, ~~p(b=1)=C_1, ~~p(b=2)=C_2,
\end{equation}
where $C_0$, $C_1$, and $C_2$ are some real constants, similarly define ellipsoids in a complex four-dimensional space. 

To maximise $p(b=i)$ subject to the normalisation constraint in Eq. \eqref{eq:ellipsenormalisation} is then equivalent to finding the largest $C_i$ for which the ellipsoid for $p(b=i)$ still shares points with the normalisation ellipsoid defined by Eq. \eqref{eq:ellipsenormalisation}.
The ellipsoids will then be tangent to each other.

In order to find the corresponding maximal values of $C_0$, $C_1$, and $C_2$, we first express all ellipsoids in the basis corresponding to the major axes of the ellipsoid for $M_0+M_1+M_2$.  Then, we rescale these axes so that the normalisation ellipsoid becomes a sphere with radius $1$ in four-dimensional complex space (the lengths of all major axes are the same). This will ``squash" the ellipsoids corresponding to $M_0$, $M_1$, and $M_2$, but they remain ellipsoids. The largest possible value for $C_i$ will then be obtained when the normalisation ellipsoid -- now a sphere -- is just contained inside the transformed ellipsoid corresponding to $M_i$, with the ellipsoid tangent to the sphere. This will be the case when the shortest major axis of the transformed ellipsoid for $M_i$ has length $1$ (the same length as the radius of the normalisation sphere).

The major axes of an ellipsoid can be found from the eigenvectors of the corresponding matrix,  and the lengths of the major axes can be found from the respective eigenvalues.
In the eigenbasis of the corresponding matrix, the equation for an ellipsoid can be written
\begin{equation}
\sum_i \lambda_i |x_i|^2=C,
\end{equation}
where $\lambda_i$ are the eigenvalues of the corresponding matrix, $x_i$ are the coordinates expressed in the eigenbasis, and $C$ is a constant. The lengths of the major axes of this ellipsoid are given by $\sqrt{C/\lambda_i}$. The shortest major axis corresponds to the largest $\lambda_i=\lambda_{max}$. If the shortest major axis has length $1$, then the largest possible value of $C$ is equal to $\lambda_{max}$. 

We will therefore need to find the eigenvalues of the transformed $M_0$, $M_1$, and $M_2$ for the corresponding squashed ellipsoids. The largest possible $p(b=i)$ Alice can obtain is then given by the largest of these eigenvalues.
The matrix $M_0+M_1+M_2$ is circulant, and hence its eigenvectors are the ``FFT (Finite Fourier Transform) vectors"
\begin{eqnarray} \label{eq:FFTvectors}
\ket{\lambda_0}&=&\frac{1}{2}(1,1,1,1)^T,~~ \ket{\lambda_1}=\frac{1}{2}(1,i,-1,-i)^T, \\
 \ket{\lambda_2}&=&\frac{1}{2}(1,-1,1,-1)^T,~~ \ket{\lambda_3}=\frac{1}{2}(1,-i,-1,i)^T.~~ ~\nonumber
\end{eqnarray}
The corresponding eigenvalues are
\begin{eqnarray}
\lambda_0&=&1+G+2{\rm Re}\, F,~~ \lambda_1=1-G+2{\rm Im}\, F,\nonumber\\
\lambda_2&=&1+G-2{\rm Re}\, F,~~ \lambda_3=1-G-2{\rm Im}\, F.
\end{eqnarray}
We now define a matrix $V$, with columns given by the FFT vectors in Eq. \eqref{eq:FFTvectors}, and a diagonal matrix 
\begin{equation}
D_{sq}=diag(\sqrt{\lambda_0}, \sqrt{\lambda_1}, \sqrt{\lambda_2}, \sqrt{\lambda_3}).
\end{equation}
It then holds that
\begin{equation}
D_{sq}^{-1}V^\dagger (M_0+M_1+M_2) V D^{-1}_{sq}=diag(1,1,1,1),
\end{equation}
that is, a $4\times 4$ identity matrix. This transformation corresponds to writing the normalisation ellipsoid in scaled coordinates where it corresponds to a sphere.
In these same coordinates, the equation for the ellipsoid corresponding to a matrix $M$ is
\begin{equation}
(\tilde\alpha^*,\tilde\beta^*,\tilde\gamma^*,\tilde\delta^*)D^{-1}_{sq} V^\dagger M V D_{sq}^{-1}
(\tilde\alpha,\tilde\beta,\tilde\gamma,\tilde\delta)^T=C,
\end{equation}
where $C$ is a constant and 
\begin{equation}
    (\tilde\alpha,\tilde\beta,\tilde\gamma,\tilde\delta)^T=D_{sq} V^\dagger(\alpha,\beta,\gamma,\delta)^T
\end{equation}
are the coordinates in the transformed basis.

The matrix $M_2$ is diagonal in the same basis as $M_0+M_1+M_2$ and the calculation is simpler in this case. The eigenvalues of the transformed matrix $D^{-1}_{sq}V^\dagger M_2 V D^{-1}_{sq}$ are
\begin{eqnarray}
\tilde\lambda_{20}&=&\frac{1/3+G}{1+G+2{\rm Re}\, F},~~~
\tilde\lambda_{21}=\frac{1/3-G}{1-G+2{\rm Im}\, F},\nonumber\\
\tilde\lambda_{22}&=&\frac{1/3+G}{1+G-2{\rm Re}\, F},~~~
\tilde\lambda_{23}=\frac{1/3-G}{1-G-2{\rm Im}\, F}, ~~
\end{eqnarray}
which are simply the eigenvalues of $M_2$ divided by the corresponding eigenvalues of $M_0+M_1+M_2$. It follows that the largest $p(b=2)$ Alice can achieve is
\begin{equation} \label{appendix:Alice cheats Bob tests p(b=2)}
p(b=2)_{\rm max}=
\begin{cases}
\frac{1/3+G}{1+G-2|{\rm Re}\, F|} &{\rm if}~~G\ge \frac{|{\rm Im}\, F|-|{\rm Re}\, F|}{2-3|{\rm Re}\, F|-3|{\rm Im}\, F|}\\
\\
\frac{1/3-G}{1-G-2|{\rm Im}\, F|}&{\rm if}~~G< \frac{|{\rm Im}\, F|-|{\rm Re}\, F|}{2-3|{\rm Re}\, F|-3|{\rm Im}\, F|}.
\end{cases}
\end{equation}

The matrices $V^\dagger M_0 V$ and $V^\dagger M_1 V$
are given by
\begin{eqnarray}
V^\dagger M_0 V &=& \begin{pmatrix}
\frac{1}{3}+{\rm Re}\, F & 0 & i {\rm Im}\, F & 0 \\
0 & \frac{1}{3} + {\rm Im}\, F & 0 & -i{\rm Re}\, F \\
-i {\rm Im}\, F & 0 & \frac{1}{3} - {\rm Re}\, F & 0 \\
0 & i {\rm Re}\, F & 0 & \frac{1}{3} - {\rm Im}\, F
\end{pmatrix}, \nonumber \\
V^\dagger M_1 V &=& \begin{pmatrix}
\frac{1}{3}+{\rm Re}\, F & 0 & -i {\rm Im}\, F & 0 \\
0 & \frac{1}{3} + {\rm Im}\, F & 0 & i{\rm Re}\, F \\
i {\rm Im}\, F & 0 & \frac{1}{3} - {\rm Re}\, F & 0 \\
0 & -i {\rm Re}\, F & 0 & \frac{1}{3} - {\rm Im}\, F
\end{pmatrix}, \nonumber \\
\end{eqnarray}
from which the matrices $D_{sq}^{-1}V^\dagger M_0 VD_{sq}^{-1}$ and $D_{sq}^{-1}V^\dagger M_1 VD_{sq}^{-1}$ can be obtained by dividing the element in position $(j,k)$ by $\sqrt{\lambda_j\lambda_k}$. Both matrices are block diagonal, which can be seen more readily if permuting, e.g., the middle two rows and columns.
The eigenvalues of $D^{-1}_{sq}V^\dagger M_0 V D^{-1}_{sq}$ are given by
\begin{eqnarray} \label{appendix:Alice cheats Bob tests p(b=0)}
&&\tilde\lambda_{00/02}=\frac{1}{(1+G)^2-4({\rm Re}\, F)^2}\left[\frac{1}{3}(1+G)-2({\rm Re}\,  F)^2\right.
\nonumber\\
&&\left.\pm\sqrt{(\frac{1}{3}+G)^2({\rm Re}\, F)^2+[(1+G)^2-4({\rm Re}\, F)^2]({\rm Im}\, F)^2}\right],\nonumber\\
&&\tilde\lambda_{01/03}=\frac{1}{(1-G)^2-4({\rm Im}\, F)^2}\left[\frac{1}{3}(1-G)-2({\rm Im}\,  F)^2\right.
\nonumber\\
&&\left.\pm\sqrt{(\frac{1}{3}-G)^2({\rm Im}\, F)^2+[(1-G)^2-4({\rm Im}\, F)^2]({\rm Re}\, F)^2}\right],\nonumber\\
\end{eqnarray}
where the $+$ sign is chosen for $\tilde\lambda_{00}$ and $\tilde\lambda_{01}$, and the $-$ sign for $\tilde\lambda_{02}$ and $\tilde\lambda_{03}$. Clearly, $\tilde\lambda_{00}$ and  $\tilde\lambda_{01}$ are the larger pair of eigenvalues, and one of these will give the largest probability Alice can achieve for $p(b=0)$. The eigenvalues for $D^{-1}_{sq}V^\dagger M_1 V D^{-1}_{sq}$ are identical, and hence Alice's cheating probability for $b=1$ is the same as for $b=0$,
\begin{equation}
p(b=0)_{\rm max}=p(b=1)_{\rm max}=\max (\tilde\lambda_{00}, \tilde\lambda_{01}).
\end{equation}
Alice's overall cheating probability is the larger of $p(b=0)_{\rm max}=p(b=1)_{\rm max}$ and $p(b=2)_{\rm max}$.

\section{Equivalence between semi-random XOT and standard XOT} \label{app:Equivalence}
 
Implementing a semi-random XOT protocol with cheating probabilities $A_{OT}$ and $B_{OT}$ allows us to realise a standard XOT protocol with the same cheating probabilities, when adding classical post-processing. We will now show that this holds true, using similar arguments as in~\cite{Amiri21, Chailloux13}, where it was shown that the variants of random, semi-random, and standard 1-2 OT are equivalent. A random version of 1-2 OT has previously already been considered in ~\cite{Crepeau88} as well.

\begin{proposition}
A semi-random XOT protocol with cheating probabilities $A_{OT}$ and $B_{OT}$ is equivalent to having a standard XOT protocol with the same cheating probabilities.
\end{proposition}
\begin{proof}
We examine both directions, i.e., constructing a semi-random XOT from a standard XOT protocol, and constructing a standard XOT protocol from a semi-random XOT protocol. That is, the situation where the parties possess means to implement standard XOT, but both of them instead wish to implement semi-random XOT, or vice versa.
\smallskip \\
\noindent \underline{Case 1:} Let $P$ be a standard XOT protocol with cheating probabilities $A_{OT}(P)$ and $B_{OT}(P)$. We can construct a semi-random XOT protocol $Q$ with the same cheating probabilities in the following way:
\begin{enumerate}
\item Alice picks $x_0, x_1 \in \{0, 1\}$ uniformly at random. Bob generates $b \in \{0, 1, 2\}$ uniformly at random (in a way so that he no longer actively chooses $b$). 
\item Alice and Bob perform the XOT protocol $P$ where Alice inputs $x_0, x_1$, and $x_2 = x_0 \oplus x_1$ and Bob inputs $b$. Let $y$ be Bob's output.
\item Alice and Bob abort in $Q$ if and only if they abort in $P$. Otherwise, the outputs of protocol $Q$ are $(b, y)$ for Bob.
\end{enumerate}
Evidently, $Q$ implements semi-random XOT if both parties follow the protocol. Furthermore, because of the way $Q$ is constructed,  Alice can cheat in $Q$ iff she can cheat in $P$, and the same for Bob cheating. Cheating probabilities for Alice and Bob are therefore equal in $P$ and $Q$, $A_{OT}(Q) = A_{OT}(P)$ and $B_{OT}(Q) = B_{OT}(P)$. 
\smallskip \\
\noindent \underline{Case 2:} Let $P$ be a semi-random XOT protocol with cheating probabilities $A_{OT}(P)$ and $B_{OT}(P)$. We can construct a standard XOT protocol $Q$ with the same cheating probabilities in the following way:
\begin{enumerate}
\item Alice has inputs $X_0, X_1$, with $X_2 = X_0 \oplus X_1$,
and Bob has input $B \in \{0, 1, 2\}$.
\item Alice and Bob perform the semi-random XOT protocol $P$ where Alice inputs $x_0, x_1$, with $x_2 = x_0 \oplus x_1$, 
whereby she chooses $x_0, x_1 \in \{0, 1\}$ uniformly at random. Let $(b, y)$ be Bob's outputs.
\item Bob sends $r = (b + B + B )\ \text{mod } 3$ to Alice. Let $x'_c = x_{(c + r) \ \text{mod }3}$ for $c \in \{0, 1, 2\}$.
\item Alice sends $(s_0, s_1)$ to Bob, whereby $s_c = x'_c \oplus X_c$ for $c \in \{0, 1\}$ and $s_2 = s_0 \oplus s_1$. Let $y' = y \oplus s_B$.
\item Alice and Bob abort in $Q$ if and only if they abort in $P$. Otherwise, the output of protocol $Q$ is $y'$ for Bob.
\end{enumerate}
If Alice and Bob are honest, then $y = x_b$. Note that $x'_B = x_{(B + r) \text{ mod }3} = x_{(B + b + B + B) \text{ mod }3} = x_b$. Hence,
\begin{equation}
y' = y \oplus s_B = x_b \oplus s_B = x'_B \oplus x'_B \oplus X_B = X_B , \nonumber
\end{equation}
i.e., $y'$ is indeed equal to $X_B$. This also holds for $B=2$ since
\begin{align*}
s_2 =& s_0 \oplus s_1 = x'_0 \oplus X_0 \oplus x'_1 \oplus X_1 \\
=& x_0 \oplus x_1 \oplus X_0 \oplus X_1 = x_2 \oplus X_2 = x'_2 \oplus X_2 .
\end{align*}
With respect to the classical post-processing described in steps 3 and 4 and security against Alice and Bob, we can conclude the following:
\begin{itemize}
\item If Alice is honest, she knows $r$ but has no information about $b$. From $r = (b + B + B) \text{ mod }3$ she can deduce that $2B = (r - b) \text{ mod }3$ but she cannot obtain any information about $B$ from this. Hence, the classical post-processing does not give an honest Alice any more information about which bit Bob has obtained.
\item If Alice is dishonest, she can correctly guess $b$ with probability $A_{OT}(P)$. She knows $r$. Since $2B = (r - b) \text{ mod }3$, guessing $2B$, equivalently guessing $B$, is equivalent to guessing $b$. Therefore, $A_{OT}(Q) = A_{OT}(P)$.
\item If Bob is honest, he knows $(s_0, s_1)$, $s_2 = s_0 \oplus s_1$, and $r$ but has no information about $x_{(b+1) \text{ mod 3}}$ and  $x_{(b+2) \text{ mod 3}}$. He cannot learn anything about the other two of Alice's bits, $X_{(B+1) \text{ mod 3}}$ and $X_{(B+2) \text{ mod 3}}$, since
\begin{align*}
 X_{(B+1) \text{ mod 3}} &= x'_{(B+1) \text{ mod 3}} \oplus s_{(B+1) \text{ mod 3}} \\
 &= x_{(B+1+r) \text{ mod 3}} \oplus s_{(B+1) \text{ mod 3}} \\
 &= x_{(b+1) \text{ mod 3}} \oplus s_{(B+1) \text{ mod 3}} , \\
 X_{(B+2) \text{ mod 3}} &= x'_{(B+2) \text{ mod 3}} \oplus s_{(B+2) \text{ mod 3}} \\
 &= x_{(B+2+r) \text{ mod 3}} \oplus s_{(B+2) \text{ mod 3}} \\
 &= x_{(b+2) \text{ mod 3}} \oplus s_{(B+2) \text{ mod 3}} .
\end{align*}
Hence, the classical post-processing does not give an honest Bob any more information about the other two bits Alice has sent. 
\item If Bob is dishonest, he can guess $x_{(b+1) \text{ mod 3}}$ and  $x_{(b+2) \text{ mod 3}}$ with probability $B_{OT}(P)$. He knows $(s_0, s_1)$, $s_2 = s_0 \oplus s_1$, and $r$. We have $s_c = x'_c \oplus X_c = x_{(c + r) \text{ mod }3} \oplus X_c$ for $c \in \{0, 1, 2\}$. Thus, $X_c = x_{(c + r) \text{ mod }3} \oplus s_c$ and, for Bob, guessing $(X_0, X_1, X_2)$ is equivalent to guessing $(x_0, x_1, x_2)$. Therefore, $B_{OT}(Q) = B_{OT}(P)$.
\end{itemize}
\end{proof}

In the classical post-processing given above, Alice needs to define her actual bit values as the bits $X_0, X_1$, with $X_2 = X_0 \oplus X_1$. The bits $x_0, x_1$, and $x_2 = x_0 \oplus x_1$ used in the semi-random XOT protocol, on the other hand, are ``dummy" values that she chooses uniformly at random. The value of $r$ will tell Alice how to permute the bits $x'_c$ before computing and sending $(s_0, s_1)$, so that Bob can learn the bit $X_B$ that he wants to learn. For example, if $r=0$, then $b=B$ and the order is fine, so Alice needs to do nothing before computing and sending $(s_0, s_1)$. If $r=1$, then $b \neq B$, and Alice needs to shift the bits one place to the left, i.e. $( x'_1, x'_2, x'_0)$, before computing and sending $(s_0, s_1)$, so that $s_0 = x'_1 \oplus X_0$ and $s_1 = x'_2 \oplus X_1$. Lastly, if $r=2$, then $b \neq B$ as well, and Alice needs to shift the bits one place to the right, i.e. $( x'_2, x'_0, x'_1)$, before computing and sending $(s_0, s_1)$, so that $s_0 = x'_2 \oplus X_0$ and $s_1 = x'_0 \oplus X_1$. In all these cases, it holds that $s_2 = s_0 \oplus s_1$, and the respective changes due to the order of the $x'_c$, for $c \in \{0, 1, 2\}$, follow. The value of $s_B$ in turn will tell Bob what to do with the value of $y$, so that he can learn his chosen bit. If $s_B=0$, then he keeps the value as it is, and if $s_B=1$, then he flips the bit value.

As mentioned in Section \ref{sec:Reversed protocol}, classical post-processing can also be added to the reversed XOT protocol. This way it is possible for Alice to choose which bit values she wants to obtain. The post-processing is straightforward and
involves only classical communication from Alice to Bob. For brevity, we outline it without full formal proofs. Suppose Alice has obtained the two bits $(x_0, x_1)$ from the reversed XOT protocol, but her desired bits are $(X_0, X_1)$. If either $x_0$ or $x_1$ is not the bit value she wants, she needs to ask Bob to flip the corresponding bit value, if he holds it. This obviously gives Bob no more information about Alice's bit values $X_0, X_1$ than what he already has about $x_0, x_1$, and also gives Alice no more information about what Bob has learnt. Defining  $t_c = x_c \oplus X_c$ for $c \in \{0,1,2\}$, Alice sends $(t_0, t_1)$ to Bob (it holds that $t_2 = x_2\oplus X_2=t_0\oplus t_1$). From the reversed XOT protocol, Bob holds the values of $b$ and bit $x_b$, and calculates $X_b=x_b\oplus t_b$ as his final bit value. Since Bob does not know $x_{(b+1) \text{ mod }3}$ and $x_{(b+2) \text{ mod }3}$, $t_{(b+1) \text{ mod }3}$ and $t_{(b+2) \text{ mod }3}$ 
do not help him at all to learn  about $X_{(b+1) \text{ mod }3}$ or $X_{(b+2) \text{ mod }3}$. He can correctly guess $X_{(b+1) \text{ mod }3}$ or $X_{(b+2) \text{ mod }3}$ with the same probability as he can guess $x_{(b+1) \text{ mod }3}$ or $x_{(b+2) \text{ mod }3}$. Thus, this classical post-processing does not increase Bob's cheating probability. It also does not increase Alice's cheating probability, since she receives no communication from Bob during the post-processing.

\section{Reworking an interactive XOT protocol into a non-interactive protocol} 
\label{app:Reworking}

Starting with the interactive XOR oblivious transfer protocol defined as protocol (3) by Kundu \textit{et al.}~\cite{Kundu22}, which uses entangled states, we show how to rework it into a non-interactive XOT protocol which requires no entanglement. In protocol (3), the sender Alice has two input bits $x_0$ and $x_1$, and the receiver Bob prepares one of three possible entangled two-qutrit states $\ket{\psi_b^+}$,
\begin{eqnarray}
\ket{\psi_0^+}&=&\frac{1}{\sqrt 2}(\ket{00}+\ket{22}),\nonumber\\
\ket{\psi_1^+}&=&\frac{1}{\sqrt 2}(\ket{11}+\ket{22}),\nonumber\\
\ket{\psi_2^+}&=&\frac{1}{\sqrt 2}(\ket{00}+\ket{11}),
\end{eqnarray}
depending on his randomly chosen input $b\in\{0,1,2\}$, which specifies whether he will learn the first bit $x_0$, the second bit $x_1$, or their XOR $x_2 = x_0\oplus x_1$. Bob then sends one of the qutrits to the sender Alice who performs a unitary operation
\begin{equation}
U_{(x_0,x_1)}=(-1)^{x_0}\ket{0}\bra{0}+(-1)^{x_1}\ket{1}\bra{1}+\ket{2}\bra{2}
\end{equation}
on it, depending on her randomly chosen input bits $(x_0, x_1)$, before sending it back to Bob. Finally, Bob performs the 2-outcome measurement $\{\ket{\psi_b^+}\bra{\psi_b^+}, {\bf 1}-\ket{\psi_b^+}\bra{\psi_b^+}\}$, and obtains either $x_0, x_1$, or $x_2$, depending on his previous choice of $b$. 
 
We can rework this interactive protocol into a non-interactive protocol, with only one quantum state transmission between Alice and Bob. That this is possible without affecting the cheating probabilities for Alice and Bob is far from evident, but it turns out that cheating probabilities do remain unchanged in this particular case. First, instead of preparing one of the three states that he prepares in protocol (3), Bob could instead prepare the entangled state
\begin{align} \label{eq:Bobstate}
\dfrac{1}{\sqrt{6}}\big[(\ket{00} +\ket{22}) \otimes \ket{0} &+ (\ket{11}+\ket{22}) \otimes \ket{1} \nonumber \\
+ (\ket{00} &+\ket{11}) \otimes \ket{2} \big] .
\end{align}
If Bob measures the last qutrit in the $\{\ket{0}, \ket{1}, \ket{2}\}$ basis, he effectively prepares one of the three states in protocol (3), with probability $1/3$ each. The only difference is that his choice of $b$ is now determined by his measurement outcome, rather than being an active choice for Bob. If starting with the state in Eq. \eqref{eq:Bobstate}, he can also just as well delay his measurement of the final qutrit to the end of the protocol. As we have shown, classical post-processing can be used to allow Bob to nevertheless make an active (but random from Alice's point of view) choice of $b$.

Whether he has measured the final qutrit or not, Bob could then send one of the two first qutrits to Alice. If Alice applies her unitary operation to one of the first two qutrits, then, depending on the values of  $x_0$ and $x_1$, the overall state becomes
\begin{align} \label{CombinationUnitary}
\ket{\phi_{x_0x_1}} = \dfrac{1}{\sqrt{6}} \big[ \big(&(-1)^{x_0}\ket{00}+\ket{22} \big) \otimes \ket{0} \nonumber \\
+ \big(&(-1)^{x_1}\ket{11}+\ket{22} \big) \otimes \ket{1} \nonumber \\
+ \big(&(-1)^{x_0}\ket{00}+(-1)^{x_1}\ket{11} \big) \otimes \ket{2} \big] .
\end{align}
Alice can now send the qutrit back to Bob after her unitary transformation, and Bob makes a measurement eliminating two out of the four possible states in order to learn either $x_0, x_1$, or $x_2$.

Since the state in Eq. \eqref{eq:Bobstate} is known to both Bob and Alice, we might ask what changes if Alice prepares the state instead of Bob. She could then apply her unitary transforms or she could straight away create one of the states in Eq. \eqref{CombinationUnitary}  (if she is honest). Apart from the fact that Bob randomly obtains either $x_0, x_1$, or $x_2 = x_0\oplus x_1$, Alice's cheating probability might then increase, since she may have additional cheating strategies available to her. Similarly, Bob's cheating probability might decrease, since he will have fewer cheating strategies at his disposal. The advantage is that there is no need for entanglement anymore, as was the case in the interactive version of this protocol. Instead of an entangled state of three qutrits, Alice could use a single quantum system; since there are four possible pure states  in Eq. \eqref{CombinationUnitary}, at most four dimensions would be needed. In this case, the state space is actually only three-dimensional. The states in Eq. \eqref{CombinationUnitary} have the same pairwise overlaps as the qutrit states in Eq. \eqref{Honest Alices States}. The reworked protocol (3) therefore becomes equivalent to our XOT protocol in Section \ref{sec:XOT protocol}. In Section \ref{sec:XOT protocol}, it was shown that this protocol has the same cheating probabilities as the interactive protocol (3) in~\cite{Kundu22}.
The price for the non-interactivity is that in the XOT protocol in Section \ref{sec:XOT protocol}, Bob cannot actively choose if he wants to receive the first bit, the second bit, or the XOR. As we show in Appendix \ref{app:Equivalence}, however, it is possible to let Bob actively choose $b$ by implementing classical post-processing.

\section{Cheating probabilities in the reversed protocol}
\label{app:ReversedCheating}

\subsection{Bob cheating in the reversed protocol}

Dishonest Bob's aim still is to learn not just one bit but any two of $x_0,\ x_1$, or $x_2 = x_0\oplus x_1$. In the reversed protocol, he wants to know exactly which of the four two-bit combinations Alice has obtained.
As in the unreversed protocol, we can consider two scenarios: one, where the receiver of the state (now Alice) tests the state, and another, where the receiver of the state does not test.

\subsubsection{Alice not testing}

When Alice does not test, Bob's optimal cheating strategy is to send Alice the eigenstate corresponding to the largest eigenvalue of one of  Alice's measurement operators. Each of the four measurement operators in Eq. \eqref{eq:Alicereversemeas} is proportional to a pure-state projector and their largest eigenvalues are all equal to $3/4$. 
Therefore, Bob's highest cheating probability is $B^\text{r}_{OT} = 3/4$, which is the same as his cheating probability in the unreversed non-interactive protocol.

\subsubsection{Alice testing}

When Alice tests, Bob needs to send a state that will pass her test. The testing is analogous to the one applied by Bob in the unreversed protocol. Alice tests a fraction of the states she receives, to see if her measurement results match Bob's declarations for this fraction of states. She aborts the protocol if there are mismatches, and otherwise continues with the rest of the protocol with the remaining states. For the same reasons as earlier, the state a dishonest Bob has to send, if he wants to pass the test every time, needs to be a superposition of the states he is supposed to send, entangled with a system he keeps. This is a state of the form
\begin{align}
\ket{\Phi^\text{r}_{\text{cheat}}} = a \ket{0}_B \otimes \ket{\phi_{x_0=0}}  &+  b \ket{1}_B \otimes \ket{\phi_{x_0=1}}  \nonumber \\
+  c \ket{2}_B \otimes \ket{\phi_{x_1=0}}  &+  d \ket{3}_B \otimes \ket{\phi_{x_1=1}}  \nonumber \\
+  e \ket{4}_B \otimes \ket{\phi_{x_2=0}}  &+  f \ket{5}_B \otimes \ket{\phi_{x_2=1}} ,
\end{align}
where $\{ \ket{0}_B, \ket{1}_B, \ket{2}_B, \ket{3}_B, \ket{4}_B, \ket{5}_B\}$ is an orthonormal basis for the system Bob keeps and $|a|^2 + |b|^2 + |c|^2 + |d|^2 + |e|^2 + |f|^2 = 1$.

After Alice has made her measurement, Bob's system on his side is prepared in one of four states, depending on whether Alice has obtained 00, 01, 11, or 10. The states he needs to distinguish between are the pure states
\begin{align} \label{Reversed States Alice cheat Bob test}
\ket{\theta_{00}} &= \dfrac{1}{\sqrt{|a|^2+|c|^2+|e|^2}} \big(a \ket{0}_B + c \ket{2}_B + e \ket{4}_B \big) , \nonumber \\
\ket{\theta_{01}} &= \dfrac{1}{\sqrt{|a|^2+|d|^2+|f|^2}} \big(a \ket{0}_B + d \ket{3}_B - f \ket{5}_B \big) , \nonumber \\
\ket{\theta_{11}} &= \dfrac{1}{\sqrt{|b|^2+|d|^2+|e|^2}} \big(b \ket{1}_B + d \ket{3}_B - e \ket{4}_B \big) , \nonumber \\
\ket{\theta_{10}} &= \dfrac{1}{\sqrt{|b|^2+|c|^2+|f|^2}} \big(b \ket{1}_B + c \ket{2}_B + f \ket{5}_B \big) ,
\end{align}
corresponding to Alice obtaining 00, 01, 11, or 10. The states occur with probabilities $(|a|^2+|c|^2+|e|^2)/2$, $(|a|^2+|d|^2+|f|^2)/2$, $(|b|^2+|d|^2+|e|^2)/2$, and $(|b|^2+|c|^2+|f|^2)/2$ for $\ket{\theta_{00}}$, $\ket{\theta_{01}}$, $\ket{\theta_{11}}$, and $\ket{\theta_{10}}$, respectively.

In general, it holds that the less equiprobable the states one needs to distinguish between are, the better, since, when one state occurs more often than the others, one can be more certain to guess correctly. Here, the issue for Bob is that if he makes the probabilities of the states in Eq. \eqref{Reversed States Alice cheat Bob test} more unequal, some pairwise overlaps become larger, i.e. the states are closer together, which makes distinguishing between them harder. Thus, we expect (and will be able to prove) that it is best for Bob to choose the constants such that the states are all equiprobable with a probability of $1/4$, for example, $a=b=c=d=e=f=1/\sqrt{6}$.
Substituting these values into Eq. \eqref{Reversed States Alice cheat Bob test}, the states' pairwise overlaps match the pairwise overlaps of the states in Eq. \eqref{Honest Alices States}. Thus, $\ket{\theta_{00}}$, $\ket{\theta_{01}}$, $\ket{\theta_{11}}$, and $\ket{\theta_{10}}$ are equivalent to the states in Eq. \eqref{Honest Alices States} and Bob's measurement is equivalent to distinguishing between the pure states 
\begin{align}
\ket{\phi_{00}} &= \dfrac{1}{\sqrt{3}}(\ket{0} + \ket{1} + \ket{2}), \nonumber \\
\ket{\phi_{01}} &= \dfrac{1}{\sqrt{3}}(\ket{0} - \ket{1} + \ket{2}), \nonumber \\
\ket{\phi_{11}} &= \dfrac{1}{\sqrt{3}}(\ket{0} - \ket{1} - \ket{2}), \nonumber \\ 
\ket{\phi_{10}} &= \dfrac{1}{\sqrt{3}}(\ket{0} + \ket{1} - \ket{2}),
\end{align}
corresponding to Alice obtaining 00, 01, 11, or 10, and where each state occurs with a probability of $1/4$.

Bob's best measurement is once again a minimum-error measurement. The square-root measurement is  optimal, as the states are equiprobable and symmetric. The measurement operators are $\Pi'_{00} = \frac{3}{4} \ket{\phi_{00}}\! \bra{\phi_{00}}$, $\Pi'_{01} = \frac{3}{4} \ket{\phi_{01}}\! \bra{\phi_{01}}$, $\Pi'_{11} = \frac{3}{4} \ket{\phi_{11}}\! \bra{\phi_{11}}$, and $\Pi'_{10} = \frac{3}{4} \ket{\phi_{10}}\! \bra{\phi_{10}}$. 
With this measurement, Bob's cheating probability $B^\text{r}_{OT}$, when Alice is testing the states he has sent to her, is
\begin{eqnarray}
B^\text{r}_{OT}&=&\dfrac{1}{4} \big[\text{ Tr}(\rho_{00} \Pi'_{00}) + \text{ Tr}(\rho_{01} \Pi'_{01}) \nonumber \\
&&+ \text{ Tr}(\rho_{11} \Pi'_{11}) + \text{ Tr}(\rho_{10} \Pi'_{10})\big] = \dfrac{3}{4}.
\end{eqnarray}
We can conclude that our choice for $a, b, c, d, e,$ and $f$ is an optimal choice, since we know that Bob can never cheat with a higher probability when Alice tests a fraction of the states Bob sends her, than he can do when Alice does not test any of his states. Since $B^\text{r}_{OT}=3/4$ for the case with no tests as well, there is no better way for Bob to choose the constants $a, b, c, d, e, f$ (there may be other choices that do just as well). The cheating probability for Bob in the reversed protocol is therefore the same as in the unreversed protocol.

\subsection{Alice cheating in reversed protocol}

As in the unreversed protocol, a dishonest Alice wants to learn whether Bob has obtained the first or the second bit, or their XOR. In this case, however, she is the receiver of the quantum state. 
Alice will have to distinguish between the three states
\begin{align}
\rho_{x_0} =& \dfrac{1}{2}\ket{\phi_{x_0=0}}\bra{\phi_{x_0=0}} +\dfrac{1}{2} \ket{\phi_{x_0=1}}\bra{\phi_{x_0=1}}  \nonumber \\
=& \dfrac{1}{2}\ket{0}\bra{0} + \dfrac{1}{2}\ket{2}\bra{2} 
, \nonumber \\
\rho_{x_1} =& \dfrac{1}{2}\ket{\phi_{x_1=0}}\bra{\phi_{x_1=0}} +\dfrac{1}{2} \ket{\phi_{x_1=1}}\bra{\phi_{x_1=1}}  \nonumber \\
=& \dfrac{1}{2}\ket{0}\bra{0} + \dfrac{1}{2}\ket{1}\bra{1} 
, \nonumber \\
\rho_{x_2} =& \dfrac{1}{2}\ket{\phi_{x_2=0}}\bra{\phi_{x_2=0}} +\dfrac{1}{2} \ket{\phi_{x_2=1}}\bra{\phi_{x_2=1}}  \nonumber \\
=& \dfrac{1}{2}\ket{1}\bra{1} + \dfrac{1}{2}\ket{2}\bra{2} 
.
\end{align}
These mixed states all have prior probability $1/3$, since Bob sends each of his six states with probability $1/6$.
One choice of optimal measurement for Alice has the measurement operators
\begin{align}
\Pi_{x_0} =& \dfrac{1}{2}\ket{0}\bra{0} + \dfrac{1}{2}\ket{2}\bra{2}  , \nonumber \\
\Pi_{x_1} =& \dfrac{1}{2}\ket{0}\bra{0} + \dfrac{1}{2}\ket{1}\bra{1}  , \nonumber \\
\Pi_{x_2} =& \dfrac{1}{2}\ket{1}\bra{1} + \dfrac{1}{2}\ket{2}\bra{2}.
\label{eq:AliceCheatREv}
\end{align}
This gives Alice a cheating probability $A^\text{r}_{OT}$ of
\begin{eqnarray}
A^\text{r}_{OT} = \dfrac{1}{3} \big[\text{ Tr}(\rho_{x_0} \Pi_{x_0}) + \text{ Tr}(\rho_{x_1} \Pi_{x_1}) + \text{ Tr}(\rho_{x_2}\Pi_{x_2})\big] = \dfrac{1}{2} , \nonumber \\
\end{eqnarray}
which is the same cheating probability as the one Alice can achieve in the unreversed protocol. Measuring in the $\ket{0}, \ket{1}, \ket{2}$ basis, with $b=i$ corresponding to $\ket{i}$, is another optimal measurement for Alice.
\medskip

\section{Experimental Data} \label{appendix:ex data}

Here, we present the measurement results of the experiments. The tables contain measured counts $C$, relative frequencies $f$, and corresponding theoretical probabilities $p_t$. The digits in parenthesis represent one standard deviation at the final decimal place.

In Table~\ref{EXP-tab-direct-honest}, we show the experimental data for the unreversed XOT protocol when both parties are honest. Alice sent states $| \phi_{00} \rangle$, $| \phi_{01} \rangle$, $| \phi_{11} \rangle$, $| \phi_{10} \rangle$ and Bob made an unambiguous quantum state elimination measurement.
\begin{widetext}

\begin{table*}[!h]
	\begin{tabular}{|l|c|r|r|r|r|r|r|}
	\multicolumn{2}{l}{} & \multicolumn{6}{c}{Bob} \\
	\cline{3-8}
	\multicolumn{2}{l|}{} & \multicolumn{1}{|c|}{$\Pi_A$} & \multicolumn{1}{|c|}{$\Pi_B$} & \multicolumn{1}{|c|}{$\Pi_C$} & \multicolumn{1}{|c|}{$\Pi_D$} & \multicolumn{1}{|c|}{$\Pi_E$} & \multicolumn{1}{|c|}{$\Pi_F$} \\	
	\multicolumn{2}{l|}{Alice} & \multicolumn{1}{|c|}{$x_0=0$} & \multicolumn{1}{|c|}{$x_0=1$} & \multicolumn{1}{|c|}{$x_1=0$} & \multicolumn{1}{|c|}{$x_1=1$} & \multicolumn{1}{|c|}{$x_2=0$} & \multicolumn{1}{|c|}{$x_2=1$} \\
	\hline
	& $C$ & 166443 & 5562   & 167526 & 719    & 167691 & 1389 \\
	$| \phi_{00} \rangle$ & $f$ & 0.3268(7) & 0.0109(1) & 0.3289(7)  & 0.00141(5) & 0.3292(7)  & 0.00273(7) \\
	& $p_t$ & 1/3 & 0 & 1/3 & 0 & 1/3 & 0 \\
	\hline
	& $C$ & 167799 & 4375   & 272    & 167383 & 1001   & 166933 \\
	$| \phi_{01} \rangle$ & $f$ & 0.3305(7) & 0.0086(1) & 0.00054(3) & 0.3296(7)  & 0.00197(6) & 0.3288(7) \\
	& $p_t$ & 1/3 & 0 & 0 & 1/3 & 0 & 1/3 \\
	\hline
	& $C$ & 4540   & 167803 & 167806 & 446    & 1189   & 168087 \\
	$| \phi_{10} \rangle$& $f$ & 0.0089(1) & 0.3291(7) & 0.3291(7)  & 0.00087(4) & 0.00233(7) & 0.3297(7) \\
	& $p_t$  & 0 & 1/3 & 1/3 & 0 & 0 & 1/3 \\
	\hline
	& $C$ & 3791   & 166615 & 317    & 166221 & 167797 & 1789  \\
	$| \phi_{11} \rangle$& $f$ & 0.0075(1) & 0.3289(7) & 0.00063(4) & 0.3282(7)  & 0.3313(7)  & 0.00353(8) \\
	& $p_t$ & 0 & 1/3 & 0 & 1/3 & 1/3 & 0 \\
	\hline
	\end{tabular}
	\caption{Measured counts $C$, relative frequencies $f$, and corresponding theoretical probabilities $p_t$ for the situation when both the parties were honest. $x_2 = x_0 \oplus x_1$.}
	\label{EXP-tab-direct-honest}
\end{table*}
\medskip
In Table~\ref{EXP-tab-direct-alice}, we show the experimental data for the case of a dishonest Alice in the unreversed XOT protocol. Alice sent states $\ket{0}$, $\ket{1}$, $\ket{2}$, while Bob honestly made an unambiguous quantum state elimination measurement.
\bigskip

\begin{table*}[!htb]
	\begin{tabular}{|l|c|r|r|r|r|r|r|}
		\multicolumn{2}{l}{} & \multicolumn{6}{c}{Bob} \\
		\cline{3-8}
		\multicolumn{2}{l|}{} & \multicolumn{1}{|c|}{$\Pi_A$} & \multicolumn{1}{|c|}{$\Pi_B$} & \multicolumn{1}{|c|}{$\Pi_C$} & \multicolumn{1}{|c|}{$\Pi_D$} & \multicolumn{1}{|c|}{$\Pi_E$} & \multicolumn{1}{|c|}{$\Pi_F$} \\
		\multicolumn{2}{l|}{Alice} & \multicolumn{1}{|c|}{$x_0=0$} & \multicolumn{1}{|c|}{$x_0=1$} & \multicolumn{1}{|c|}{$x_1=0$} & \multicolumn{1}{|c|}{$x_1=1$} & \multicolumn{1}{|c|}{$x_2=0$} & \multicolumn{1}{|c|}{$x_2=1$} \\
		\hline
		& $C$ & 126264      & 135006     & 124653     & 121434     & 29         & 30  \\
		$|0 \rangle$ & $f$ & 0.2488(6)   & 0.2661(6)  & 0.2457(6)  & 0.2393(6)  & 0.00006(1) & 0.00006(1) \\
		$x_0, x_1$ & $p_t$ & 1/4 & 1/4 & 1/4 & 1/4 & 0 & 0 \\
		\hline
		& $C$ & 10          & 189        & 127189     & 129235     & 131522     & 121722 \\
		$| 1 \rangle$ & $f$ & 0.000020(6) & 0.00037(3) & 0.2495(6)  & 0.2535(6)  & 0.2580(6)  & 0.2387(6)\\
		$x_1, x_2$ & $p_t$ & 0 & 0 & 1/4 & 1/4 & 1/4 & 1/4 \\		
		\hline
		& $C$ & 130304      & 124349     & 93         & 26         & 119256     & 132601 \\
		$|2 \rangle$& $f$ & 0.2572(6)   & 0.2454(6)  & 0.00018(2) & 0.00005(1) & 0.2354(6)  & 0.2617(6) \\
		$x_0, x_2$ & $p_t$ & 1/4 & 1/4 & 0 & 0 & 1/4 & 1/4 \\
		\hline
	\end{tabular}
	\caption{Measured counts $C$, relative frequencies $f$, and corresponding theoretical probabilities $p_t$ for the situation when Alice is cheating. $x_2 = x_0 \oplus x_1$.}
	\label{EXP-tab-direct-alice}
\end{table*}
\end{widetext}

In Table~\ref{EXP-tab-direct-bob}, we show the experimental data for the case of a dishonest Bob in the unreversed XOT protocol. While Alice honestly sent the correct states, Bob applied the square-root measurement. In fact, it also shows the experimental data for the reversed XOT protocol with a dishonest Bob, only interchanging the sender and receiver roles, see names in parentheses.

In Table~\ref{EXP-tab-reversed-honest}, we show the experimental data for the reversed XOT protocol when both parties are honest. Bob sent states $| \phi_{x_0=0} \rangle$, $| \phi_{x_0=1} \rangle$, $| \phi_{x_1=0} \rangle$, $| \phi_{x_1=1} \rangle$, $| \phi_{x_2=0} \rangle$, $| \phi_{x_2=1} \rangle$ and Alice performed a POVM measurement.

In Table~\ref{EXP-tab-reversed-alice}, we show the experimental data for the case of a dishonest Alice in the reversed XOT protocol. While Bob honestly sent the correct states, Alice performed a projective measurement and classical post-processing.

\begin{table}[!htb]
	\begin{tabular}{|l|c|r|r|r|r|}
		\multicolumn{2}{l}{Alice} & \multicolumn{4}{c}{Bob (Alice)} \\
		\cline{3-6}	
		\multicolumn{2}{l|}{(Bob)} & \multicolumn{1}{|c|}{$\Pi_{00}$} & \multicolumn{1}{|c|}{$\Pi_{01}$} & \multicolumn{1}{|c|}{$\Pi_{10}$} & \multicolumn{1}{|c|}{$\Pi_{11}$} \\
		\hline
		& $C$ & 377482    & 41178     & 38173     & 43299 \\
		$| \phi_{00} \rangle$ & $f$ & 0.7547(6) & 0.0823(4) & 0.0763(4) & 0.0866(4) \\
		& $p_t$ & 3/4 & 1/12 & 1/12 & 1/12 \\
		\hline
		& $C$ & 40908     & 359828    & 52461     & 41808 \\
		$| \phi_{01} \rangle$ & $f$ & 0.0826(4) & 0.7268(6) & 0.1060(4) & 0.0844(4) \\
		& $p_t$ & 1/12 & 3/4 & 1/12 & 1/12 \\
		\hline
		& $C$ & 41904     & 39478     & 378828    & 41595 \\
		$| \phi_{10} \rangle$& $f$ & 0.0835(4) & 0.0787(4) & 0.7548(6) & 0.0829(4) \\
		& $p_t$  & 1/12 & 1/12 & 3/4 & 1/12 \\
		\hline
		& $C$ & 50901     & 42306     & 38995     & 368643 \\
		$| \phi_{11} \rangle$& $f$ & 0.1016(4) & 0.0845(4) & 0.0779(4) & 0.7360(6) \\
		& $p_t$ & 1/12 & 1/12 & 1/12 & 3/4 \\
		\hline
	\end{tabular}
	\caption{Measured counts $C$, relative frequencies $f$, and corresponding theoretical probabilities $p_t$ for the situation when Bob is cheating. These results also correspond to the reversed protocol with cheating Bob -- then the roles of sender and receiver are swapped (see the names in the parentheses)}
	\label{EXP-tab-direct-bob}
\end{table}

\newpage
\begin{table}[!htb]
	\begin{tabular}{|l|c|r|r|r|r|}
		\multicolumn{2}{l}{} & \multicolumn{4}{c}{Alice} \\
		\cline{3-6}	
		\multicolumn{2}{l|}{Bob} & \multicolumn{1}{|c|}{$\Pi_{00}$} & \multicolumn{1}{|c|}{$\Pi_{01}$} & \multicolumn{1}{|c|}{$\Pi_{10}$} & \multicolumn{1}{|c|}{$\Pi_{11}$} \\
		\hline
		& $C$ & 249402     & 239442     & 1636       & 1806 \\
		$| \phi_{x_0=0} \rangle$ & $f$ & 0.5066(7)  & 0.4864(7)  & 0.00332(8) & 0.00367(9) \\
		& $p_t$ & 1/2 & 1/2 & 0 & 0 \\
		\hline
		& $C$ & 3028       & 762        & 249215     & 246373 \\
		$| \phi_{x_0=1} \rangle$ & $f$ & 0.0061(1)  & 0.00153(6) & 0.4991(7)  & 0.4934(7) \\
		& $p_t$ & 0 & 0 & 1/2 & 1/2 \\
		\hline
		& $C$ & 249097     & 802        & 246042     & 1069 \\
		$| \phi_{x_1=0} \rangle$& $f$ & 0.5012(7)  & 0.00161(6) & 0.4950(7)  & 0.00215(7) \\
		& $p_t$  & 1/2 & 0 & 1/2 & 0 \\
		\hline
		& $C$ & 1019       & 241863     & 1840       & 246310  \\
		$| \phi_{x_1=1} \rangle$& $f$ & 0.00208(6) & 0.4926(7)  & 0.00375(9) & 0.5016(7) \\
		& $p_t$ & 0 & 1/2 & 0 & 1/2 \\
		\hline
		& $C$ & 255968     & 38         & 301        & 249572 \\
		$| \phi_{x_2=0} \rangle$& $f$ & 0.5060(7)  & 0.00008(1) & 0.00060(3) & 0.4933(7) \\
		& $p_t$  & 1/2 & 0 & 0 & 1/2 \\
		\hline
		& $C$ & 29         & 237407     & 264287     & 213\\
		$| \phi_{x_2=1} \rangle$& $f$ & 0.00006(1) & 0.4730(7)  & 0.5265(7)  & 0.00042(3) \\
		& $p_t$ & 0 & 1/2 & 1/2 & 0 \\		
		\hline
	\end{tabular}
	\caption{Reversed protocol. Measured counts $C$, relative frequencies $f$, and corresponding theoretical probabilities $p_t$ for the situation when both the parties were honest. $x_2 = x_0 \oplus x_1$.}
	\label{EXP-tab-reversed-honest}
\end{table}

\onecolumngrid
\begin{table}[!htb]
	\begin{tabular}{|l|c|r|r|r|}
		\multicolumn{2}{l}{} & \multicolumn{3}{c}{Alice} \\
		\cline{3-5}	
		\multicolumn{2}{l}{} & \multicolumn{1}{|c|}{$|0\rangle\!\langle0|$} & \multicolumn{1}{|c|}{$|1\rangle\!\langle1|$} & \multicolumn{1}{|c|}{$|2\rangle\!\langle2|$} \\
		\multicolumn{2}{l|}{Bob} & \multicolumn{1}{|c|}{$x_0, x_1$} & \multicolumn{1}{|c|}{$x_1, x_2$} & \multicolumn{1}{|c|}{$x_0, x_2$} \\
		\hline
		& $C$ & 266828     & 23          & 260337 \\
		$| \phi_{x_0=0} \rangle$ & $f$ & 0.5061(7)  & 0.000044(9) &  0.4938(7)\\
		& $p_t$ & 1/2 & 0 & 1/2 \\
		\hline
		& $C$ & 266040     & 13  & 261456\\
		$| \phi_{x_0=1} \rangle$ & $f$ & 0.5043(7)  & 0.000025(7) & 0.4956(7) \\
		& $p_t$ & 1/2 & 0 & 1/2 \\
		\hline
		& $C$ & 264336     & 255114 & 172 \\
		$| \phi_{x_1=0} \rangle$& $f$ & 0.5087(7)  & 0.4910(7) & 0.00033(3)\\
		& $p_t$  & 1/2 & 1/2 & 0 \\
		\hline
		& $C$ & 267393     & 255628 & 151\\
		$| \phi_{x_1=1} \rangle$& $f$ & 0.5111(7)  & 0.4886(7) & 0.00029(2)\\
		& $p_t$ & 1/2 & 1/2 & 0 \\
		\hline
		& $C$ & 1240       & 257057 & 262665 \\
		$| \phi_{x_2=0} \rangle$& $f$ & 0.00238(7) & 0.4934(7) & 0.5042(7) \\
		& $p_t$  & 0 & 1/2 & 1/2 \\
		\hline
		& $C$    & 1192       & 254941 & 262185 \\
		$| \phi_{x_2=1} \rangle$& $f$ & 0.00230(7) & 0.4919(7) & 0.5058(7)\\
		& $p_t$ & 0 & 1/2 & 1/2 \\		
		\hline
	\end{tabular}
	\caption{Reversed protocol. Measured counts $C$, relative frequencies $f$, and corresponding theoretical probabilities $p_t$ for the situation when Alice was cheating. $x_2 = x_0 \oplus x_1$.}
	\label{EXP-tab-reversed-alice}
\end{table}

\end{document}